\documentclass[11pt]{article}
\usepackage[utf8]{inputenc}
\usepackage[margin= 1.1 in]{geometry}
\usepackage{amsmath}
\usepackage[normalem]{ulem}
\usepackage{fontaxes}
\usepackage{trimspaces}
\usepackage{setspace}
\usepackage{inconsolata}
\usepackage{libertine}
\usepackage{todonotes}
\usepackage[T1]{fontenc}
\usepackage{amsmath, amssymb}
\usepackage{xargs}
\usepackage{mathtools}
\usepackage{amsthm}
\usepackage{thm-restate}
\usepackage{stmaryrd}
\usepackage{natbib}
\usepackage{graphicx}
\usepackage{hyperref}
\usepackage{xcolor}
\usepackage{comment}
\usepackage{todonotes}
\usepackage{hyperref}
\usepackage{xspace}

\makeatletter
\let\oldhypertarget\hypertarget
\renewcommand{\hypertarget}[2]{%
	\Hy@raisedlink{\oldhypertarget{#1}{}}
	#2%
	\protected@write\@mainaux{}{%
		\string\expandafter\string\gdef
		\string\csname\string\detokenize{#1}\string\endcsname{#2}%
	}%
}
\newcommand{\mylink}[1]{%
	\hyperlink{#1}{\csname #1\endcsname}%
}
\makeatother

\usepackage{stmaryrd}
\usepackage{textcomp}

\newcommand{\linkdest}[2]{}

\usepackage{enumitem}
\usepackage{framed}
\usepackage[capitalise]{cleveref}
\usepackage[framemethod=TikZ]{mdframed}

\usepackage[noend]{algorithmic}
\usepackage[boxed]{algorithm}

\usepackage{multirow}
\usepackage{tabularx}

\usepackage{mathtools}

\newcommand{\algcomment}[1]{\colorbox{black!10}{#1}}

\newcommand{\TSP}{\textsc{TSP}\xspace}

\newtheorem{theorem}{Theorem}[section]
\newtheorem{definition}[theorem]{Definition}
\newtheorem{lemma}[theorem]{Lemma}
\newtheorem{corollary}[theorem]{Corollary}
\newtheorem{claim}[theorem]{Claim}

\newtheorem{observation}[theorem]{Observation}
\newtheorem{conjecture}[theorem]{Conjecture}

\newtheorem{problem}{Open Problem}

\newcommand{\poly}{\mathtt{poly}}

\newcommand{\cA}{\mathcal{A}}
\newcommand{\cB}{\mathcal{B}}

\newcommand{\cF}{\mathcal{F}}

\newcommand{\cP}{\mathcal{P}}

\newcommand{\cT}{\mathcal{T}}

\newcommand{\cX}{\mathcal{X}}
\newcommand{\cY}{\mathcal{Y}}

\title{Improved Space-Time Tradeoffs for Permutation Problems via Extremal Combinatorics\footnote{All authors are supported by the project COALESCE that has received funding from the European Research Council (ERC), grant agreement No 853234.}}

\author{
  Afrouz Jabal Ameli\thanks{
    Department of Information and Computing Sciences, Utrecht University, The Netherlands.\\
    Email addresses:
    \texttt{a.jabalameli@uu.nl}, \texttt{j.nederlof@uu.nl}, \texttt{s.wang5@uu.nl}
  }
  \and
  Jesper Nederlof\footnotemark[2]
  \and
  Shengzhe Wang\footnotemark[2]
}
\date{}
\begin{document}
\maketitle
\begin{abstract}
	We provide improved space-time tradeoffs for permutation problems over additively idempotent semi-rings. In particular, there is an algorithm for the \textsc{Traveling Salesperson Problem} that solves $N$-vertex instances using space $S$ and time $T$ where $S\cdot T \leq 3.1861^{N}$. This improves a previous work by Koivisto and Parviainen~[SODA'10] where $S\cdot T \leq 3.9271^N$, and overcomes a barrier they identified, as their bound was shown to be optimal within their framework.

    To get our results, we introduce a new parameter of a set system that we call the \emph{chain efficiency}. This relates the number of maximal chains contained in the set system with the cardinality of the system. We show that set systems of high efficiency imply efficient space-time tradeoffs for permutation problems, and give constructions of set systems with high chain efficiency, disproving a conjecture by Johnson, Leader and Russell~[Comb. Probab. Comput.'15].
\end{abstract}

\section{Introduction}
The area of \emph{exact exponential time algorithms} (or \emph{fine-grained complexity of NP-complete problems}) asks for the most fundamental NP-complete problems how efficiently they can be solved exactly in the worst case. While early papers initiated this endeavor already in the 20th century~\cite{Bellman62,HeldK61,HorowitzS74}, the area flourished  around 2010 thanks to influential surveys by Woeginger~\cite{Woeginger04,Woeginger08}. There is also a textbook on the topic~\cite{FominK10}, and several more recent surveys~\cite{FominK13,Nederlof26}.

Besides the main goal of fast running times there has also consistently been major emphasis on designing algorithms with minimum space usage. Space usage is already prevalent in (even the title of) the initial survey~\cite{Woeginger04}, and the space usage of the aforementioned classical results was improved already many years ago~\cite{SchroeppelS81,Karp82}, while further progress continues to be made for certain problems~\cite{BansalGN018,BelovaCKM24,NederlofW21}. 

One of the most central computational problems studied in the area of exact exponential time algorithms is the~\textsc{Traveling Salesperson Problem (TSP)}, and it was already a topic of very early work~\cite{Bellman62,HeldK61} that solves $N$-vertex instances in $O^*(2^N)$ time.\footnote{The $O^*(\cdot)$ notation omits factors polynomial in the input size.} It poses the following very ambitious question:
\begin{problem}[Problem 6 in~\cite{Woeginger04}]\label{op}
\ 
\begin{itemize}
\item[] \textbf{(a)} Find an algorithm that solves $N$-vertex \TSP in $O^*(1.99^N)$ time,
\item[] \textbf{(b)} Find an algorithm that solves $N$-vertex \TSP in $O^*(2^N)$ time and polynomial space.
\end{itemize}
\end{problem}
While some progress on Open Problem~\ref{op}\textbf{(a)} in the unweighted and bipartite cases has been made~\cite{Bjorklund14,Nederlof20}, it still remains wide open. For Open Problem~\ref{op}\textbf{(b)}, \cite{GurevichS87} gave an $O^*(4^N)$ time and poly space algorithm\footnote{Many papers cite the algorithm of~\cite{GurevichS87} as running in $O^*(4^N N^{O(\log N)})$, but it seems folklore knowledge that it can be improved relatively easily to $O^*(4^N)$ time and polynomial space with a simple trick from~\cite{dividecolor}. For completeness we provide the simple algorithm in Appendix~\ref{sec:gs}.}, and it is a well-known open question whether this can be improved to even $O^*(3.99^N)$ time and polynomial space.

\paragraph{Space-Time Tradeoffs.}
Instead of requiring polynomial space as done in Open Problem~\ref{op}\textbf{(b)}, one can also study \emph{time-space trade-offs}: For every $S=S(N)$, what is the fastest running time $T=T(N)$ of any algorithm that solves $N$-vertex \TSP using only $S$ space? Such space-time tradeoffs are studied for several other NP-complete problems, like for example the \textsc{Subset Sum} problem~\cite{AustrinKKM13,DinurDKS12}. 

The algorithms of~\cite{Bellman62,HeldK61} and~\cite{GurevichS87} can be mixed to get, for every $i=0,\ldots,\lceil \log_2 N\rceil$ an algorithm that runs in $S=O^*(2^{N/2^i})$ space and $T=O^*(2^{(2-1/2^i)N})$ time (confer~\cite{FominK10}) and one may expect that improving over the \emph{time-space product} $S\cdot T = O^{*}(4^{N})$ is not much easier than solving Open Problem~\ref{op}.

However, perhaps somewhat surprisingly, Koivisto and Parviainen~\cite{KoivistoP10} were able to get a better trade-off for a general family of computational problems that includes \TSP. To define their general setup, let $S_n$ be the set of all permutations from $\{1,\ldots,n\}$ to $\{1,\ldots,n\}$, fix a semi-ring\footnote{Recall a \emph{semi-ring} is a set equipped with operations $\oplus$ and $\otimes$ such that addition is associative and commutative with an identity, multiplication is associative with an identity, and multiplication distributes over addition. Two important semi-rings are the \emph{Boolean semi-ring} (set $\{0,1\}$; operations $\wedge$ and $\vee$) and the \emph{Min-Sum semi-ring} (set $\mathcal{N} \cup \{\infty\}$; operations $\min$, $+$).} with addition $\oplus$ and multiplication $\otimes$, and consider the following:

\begin{definition}[\cite{KoivistoP10}]
A \emph{permutation problem of degree $d$ over a semi-ring with addition $\oplus$ and multiplication $\otimes$} is the task of evaluating $\oplus_{\sigma \in S_N} f(\sigma)$, where we can write
\[
    f(\sigma) = \bigotimes_{j=1}^N f_j(\{\sigma_1,\ldots,\sigma_j\},\sigma_{\max\{1,j-d+1\}},\ldots,\sigma_{j-1},\sigma_j)
\]
for some functions $f_1,\ldots,f_N$.
\end{definition}
As observed in~\cite{KoivistoP10}, permutation problems of bounded degree (i.e. $d=O(1)$) model many computational problems, including \TSP, \textsc{Directed Feedback Arc Set}, and~\textsc{Cutwidth}.
By generalizing the approaches from respectively~\cite{Bellman62,HeldK61} and~\cite{GurevichS87}, it can be shown relatively easily with dynamic programming over subsets of $\{1,\ldots,N\}$ that any permutation problem can be solved in time $2^N$ and space $2^N$, and respectively in $4^N$ time and $O^*(1)$ space. The following result implies the mentioned improved space-time tradeoff for \TSP:
\begin{theorem}[\cite{KoivistoP10}]\label{lem:kk}
Any permutation problem of bounded degree over a semi-ring can be solved using $S$ space and time $T$, for some $S$ and $T$ satisfying $S\cdot T \leq 3.9271^N$. 
\end{theorem}
Theorem~\ref{lem:kk} is proven in the following very elegant way: Let $P$ be a partially ordered set\footnote{Recall a \emph{poset} (short for partially ordered set) is a set equipped with a partial order $\preceq$ which is a binary relation that is reflexive, antisymmetric and transitive. We describe a poset $P$ with its \emph{Hasse diagram} which is a digraph $D$ with an arc $(u,v)$  whenever $u \preceq v$ and there is no $w$ distinct from $u$ and $v$ such that $u \preceq w\preceq v$.} with $\alpha(P)$ ideals.\footnote{An \emph{ideal} is a set of elements $X$ such that $u \preceq v$ and $v \in X$ imply that $u \in X$.} By restricting the "dynamic programming over subsets" algorithm~\cite{Bellman62,HeldK61} one can compute the quantity $\oplus_{\sigma \in L(P)} f(\sigma)$ in $O^*(\alpha(P))$ time and space, where $L(P)$ denotes the set of linear extensions of $P$. Then it is shown that there is a family $\cP$ of partial orders such that $\{L(P) : P \in \cP\}$ partitions $S_n$ and yet $\alpha(P)$ is small, leading to an algorithm that computes $\oplus_{\sigma \in S_n} f(\sigma)$ in $\sum_{P \in \cP} \alpha(P)$ time and $\max_{P \in \cP}\alpha(P)$ space and hence an algorithm with time-space product
\[
 \theta(\cP) := \left(\sum_{P \in \cP} \alpha(P)\right) \left(\max_{P \in \cP}\alpha(P) \right).
\]
The families $\cP$ in~\cite{KoivistoP10} were obtained by combining \emph{bucket orders}, which are partially ordered sets $P$ whose elements can be partitioned in sets $B_1,\ldots,B_k$ such that $(x,y) \in P$ for all $x \in B_{i}$ and $y \in B_{i+1}$. In particular, the authors used the bucket order with $k=2$ and $|B_1|=|B_2|=13$ and obtained a family $\cP$ of size $|\cP|=\tbinom{26}{13}$ with $\alpha(P)=2^{14}-1$ for all $P \in \cP$ and hence $\theta(\cP) = \tbinom{26}{13}(2^{14}-1)^2\approx 3.9271^{26}$. Additionally, they showed that this is optimal within their paradigm of combining bucket orders.

\paragraph{Improved Space-Time Tradeoffs over Idempotent Semi-rings.}
Our starting observation is that, for permutation problems over \emph{idempotent} semi-rings\footnote{Idempotent semi-rings are semi-rings with idempotent addition (i.e. $a \oplus a =a$ for every $a$). The aforementioned Boolean semi-ring and Min-Sum semi-ring are idempotent, and hence we do not lose any of the algorithmic applications mentioned in~\cite{KoivistoP10} with this restriction.}, the above still holds when $\cP$ is a family of posets such that the sets $\{L(P): P \in \cP\}$ \emph{cover} $S_n$ (i.e., every permutation in $S_n$ is a linear extension of a poset in $\cP$). This weaker requirement actually allows us to obtain the desired poset families from any fixed poset with an easy probabilistic argument: If $P$ is an $n$-element poset with $\lambda(P)$ linear extensions, then by the probabilistic method there exists a family $\cP$ of size $\frac{n!}{\lambda(P)}n \ln n$ such that, for any $P' \in \cP$,  $\alpha(P')=\alpha(P)$ and $S_n = \cup_{P \in \cP}L(P)$: Indeed, we can just take $|\cP|$ copies of $P$ whose elements are randomly relabeled and argue that with non-zero probability every permutation is a linear extension of an element of $\cP$.\footnote{The original paper~\cite{KoivistoP10} also obtained their family by such relabelings, but in a more restricted manner.}

Fixing the space usage to $m$ and observing that the number of ideals of a poset equals the number of antichains, this immediately brings our attention to the following clean question in extremal combinatorics posed by Johnson, Leader and Russell~\cite{JohnsonLR15}:

\begin{problem}[Question 8 in~\cite{JohnsonLR15}]\label{opposet}
What is the maximum number of linear extensions of a poset on $n$ elements that contains at most $m$ antichains?
\end{problem}
Somewhat surprisingly, this problem has been hardly studied in the literature. The authors of~\cite{JohnsonLR15} conjecture that extremal examples for Open Problem~\ref{opposet} basically coincide with the bucket orders from~\cite{KoivistoP10}.
This suggests that bucket orders are optimal for our purposes.

\subsection*{Our results}
As our first result, we further generalize the method from~\cite{KoivistoP10} to not only work for set systems formed by the set of ideals of a bucket order poset, but for \emph{any} set system. The parameter of interest of the set system resembles the number of linear extensions of a poset and was also already defined by Johnson, Leader and Russell~\cite{JohnsonLR15}:

\begin{definition}[Maximal chains~\cite{JohnsonLR15}]
If $\cA \subseteq 2^{[n]}$, a \emph{maximal chain} is a sequence $A_0,\ldots,A_n \in \cA$ such that $A_i \subset A_{i+1}$ for all $i=0,\ldots,n-1$. We denote $c(\cA)$ for the number of maximal chains of $\cA$.
\end{definition}
Note that $A_0=\emptyset$ and $A_n=[n]$, whenever $A_0,\ldots,A_n$ is a maximal chain. The following parameter resembles the parameter $\theta(P)$ and quantifies how many maximal chains $\cA$ has in comparison to its size and universe-size:
\begin{definition}[Chain efficiency of a set system]\label{def:chaineff}
	Let $\cA \subseteq 2^{[n]}$ be a set system on $n$ elements. Define the \emph{chain efficiency} (or \emph{efficiency} for short) $\eta(\cA)$ of $\cA$ as follows:
	\[
		\eta(\cA) := \left(\frac{c(\cA)}{|\cA|^2\cdot n!}\right)^{1/n}.
	\]
\end{definition}

By generalizing the aforementioned ``dynamic programming over subsets'' method (originally from~\cite{Bellman62,HeldK61}) and the application of the probabilistic method (derandomized in a relatively standard way), we show the following:
\begin{restatable}{lemma}{lemAlgo}\label{lem:algo}
    Let $\cA \subseteq 2^{[n]}$. Then there is an algorithm that solves permutation problems of bounded degree on instances of size $N$ over idempotent semi-rings in space $S$ and time $T$ with $S\cdot T = \eta(\cA)^{-N}N^{O(\sqrt{N})}$, assuming $n = O\left(\sqrt{N}\right)$.
\end{restatable}
Note that the sub-exponential factor $N^{O(\sqrt{N})}$ can always be omitted in this paper when applying Lemma~\ref{lem:algo} since the constant $\eta(\cA)$ is often rounded down, which suppresses sub-exponential factors.

By the following observation, Lemma~\ref{lem:algo} indeed generalizes the earlier bucket order-based and poset-based paradigms:
\begin{observation}\label{obs:posetEfficiency}
If $P$ is a poset of $n$ elements, $\alpha(P)$ ideals and $\lambda(P)$ linear extensions, then the set of ideals of $P$ has $\lambda(P)$ maximal chains and hence has chain efficiency $\left(\frac{\lambda(P)}{\alpha(P)^2 \cdot n!}\right)^{1/n}$.
\end{observation}
This observation also suggests the following natural definition:
\begin{definition}[Efficiency of a Poset]
The efficiency $\eta(P)$ of a poset $P$ denotes $\left(\frac{\lambda(P)}{\alpha(P)^2 \cdot n!}\right)^{1/n}$.
\end{definition}
Note that in Lemma~\ref{lem:algo}, family $\cA$ does not need to be known to the algorithm, so the question about the best time-space tradeoffs within our new paradigm is equivalent to asking what is the most efficient set system. While we are not able to determine this exactly, we give a poset (and hence a set system) that is significantly more efficient than the bucket order-based construction of~\cite{KoivistoP10}. To describe our construction, we define a \emph{bipartite poset} to be a poset whose Hasse Diagram is bipartite.
\begin{restatable}{lemma}{lemConstruction}\label{lem:construction}
    Let $P$ be a bipartite poset with parts $X=\{x_0,\ldots,x_{28}\}$ and $Y=\{y_0,\ldots,y_{28}\}$ such that $(x_i,y_j) \in P$ if and only if $(i -j) \mod 29 \in \{0,1,3,6,10,15\}$. Then $\eta(P) > 1/3.7492$. 
\end{restatable}
Since this is more efficient than the best bucket order-based construction, our construction disproves the conjecture of~\cite{JohnsonLR15}.

Departing from posets and making use of the generality of arbitrary set systems, we are able to improve the efficiency much further:

\begin{lemma}\label{lem:setSystemConstruction}
    There exists a set system $\cA$ with $\eta(\cA) > 1/3.1861$.
\end{lemma}
By combining Lemma~\ref{lem:algo} and Lemma~\ref{lem:setSystemConstruction}, we obtain our main theorem:
\begin{theorem}[Main theorem]
Any permutation problem of bounded degree over an idempotent semi-ring can be solved using $S$ space and time $T$, for some $S$ and $T$ satisfying $S\cdot T \leq 3.1861^N$. 
\end{theorem}
We also show various limitations of our frameworks. Most specifically, within the family of bipartite posets that are \emph{regular} (i.e. each vertex has the same number of incident arcs) we show the following:
\begin{restatable}{lemma}{bipreg}\label{lem:bipreg}
    If \(P\) is a regular bipartite poset, then $\eta(P)< 1/3.6$.
\end{restatable}
This lemma is obtained by combining an entropy-based upper bound of the number of linear extensions by Brightwell and Tetali~\cite{Brightwell2023} with a fairly direct lower bound on the number of ideals. Since Brightwell and Tetali were able to extend their bound to (specific, but non-trivial) non-bipartite posets, we believe our method may also be useful for obtaining strong upper bounds on the efficiency of general posets.

More generally, we show that within our most general setting of arbitrary set systems, it is not possible to get an algorithm with $S \cdot T < 3.015^{N}$:

\begin{restatable}{lemma}{imprb}\label{lem:imprb}
    If $\cA$ is a set system, then $\eta(\cA) < 1/3.015$.
\end{restatable}

To prove Lemma~\ref{lem:imprb}, we first provide an easier $\eta(\cA) \leq 1/3$ upper bound, which can be obtained by encoding the maximal chain $A_0 \subset A_1 \subset \ldots \subset A_n$ with $A_{n/3},A_{2n/3}$ and $3$ permutations of $S_{n/3}$. See Lemma~\ref{lem:basicUpperbound} for details. To improve this bound, we observe that, if instead we aim to encode the maximal chain with $A_{n/3},A_{n/2},A_{2n/3}$ and $2$ permutations of $S_{n/3}$ and $2$ permutations of $S_{n/6}$ then this can be done with $A_{n/2}$ and $B := A_{n/3} \cup ([n] \setminus A_{2n/3})$. Hence, if $B \in \cA$ as well this leads to an improvement. We use an isoperimetric inequality (Theorem~\ref{thm:ff} by Frankl and F\"uredi~\cite{FranklFuredi1981}) to show that often there is such a $B$ that is close to a set in $\cA$. We believe our proof strategy has the potential for even higher lower bounds and it also supports the natural line of attack of finding efficient set systems via isoperimetric inequalities.

\subsection*{Related Work}
\paragraph{Ideals versus Linear Extensions.}
As mentioned, there appear to be only very few papers that study the trade-off between the number of ideals (or equivalently, anti-chains) and the number of linear extensions of a poset. Kahn and Kim~\cite{kahn1992entropy} relate the number of linear extensions of a poset $P$ with the graph entropy of the comparability graph (which is about the entropy of a distribution over anti-chains of $P$).

\paragraph{Set Families with Many (Maximal) Chains.}
There seems to be more literature about the maximum number of chains in set systems. In particular, a research line started by Alon and Frankl~\cite{AlonF85} (see also e.g.~\cite{hunter2024disjoint}) considers the maximum number of chains of bounded length in set systems.

As suggested also by Johnson, Leader and Russell~\cite{JohnsonLR15}, Open Problem~\ref{opposet} seems similar to related problems amenable to compression techniques, and indeed they observe that one can assume that the set family $\cA$ is left-compressed. 

\paragraph{Counting Linear Extensions.}
For the computational part of the paper, we need to compute the number of linear extensions and ideals of fixed posets. Both these problems are $\#P$-complete, even on posets of height $2$ (see~\cite{DittmerP20} for a proof of $\#P$-hardness, and counting ideals of bipartite posets coincides with counting independent sets in bipartite graphs which is well-known to be $\#P$-complete). The problem of counting linear extensions of an $n$-element poset can be solved in $O^*(2^n)$ time~\cite{KangasHNK16} or in $n^{O(t)}$ time if $t$ is the treewidth of the graph. There are also polynomial time approximation schemes and algorithms that run well on many instances, see e.g.~\cite{TalvitieK24} for more details.

\paragraph{Note on Parallel Work.} We recently learned that an improved space-time tradeoff for TSP has
also been obtained independently and concurrently in~\cite{Dallant26}. Their work appears on arXiv simultaneously with ours and achieved $S\cdot T \leq 3.572^N$. We first achieved $S\cdot T\leq 3.7493^N$ in our first arXiv version and, after reading~\cite{Dallant26}, improved our bound to $S\cdot T \leq 3.1861^N$, written in Lemma~\ref{lem:setSystemConstruction}. We also learned from the authors of~\cite{Dallant26} that they recently found bounds that are optimal in some precise sense.

\subsection*{Organization}
In Section~\ref{sec:algo} we prove Lemma~\ref{lem:algo}, in Section~\ref{sec:con} we prove Lemma~\ref{lem:construction} and disprove the conjectures of~\cite{JohnsonLR15}. In Section~\ref{sec:upp} we provide a barrier beyond which our paradigm cannot improve. In Section~\ref{sec:conc} we provide concluding remarks and avenues for further research.
In Appendix~\ref{sec:gs} we provide the details for the (folklore) $O^*(4^N)$ time polynomial space algorithm for \TSP. In Appendix~\ref{sec:MissingProofs} we provide postponed technical proofs towards Lemma~\ref{lem:bipreg}.
In Appendix~\ref{sec:AppendixImplementation} we provide more details on the implementations of our computational results.

\section{From Chain Efficiency to Space-Time Tradeoffs}
\label{sec:algo}
In this section, we show that a set system with high chain efficiency implies a space and time efficient algorithm for \TSP and other permutation problems, generalizing the bucket order-based observation from~\cite{KoivistoP10}:

\lemAlgo*

We first observe that a set system on a small universe can be extended to one on a larger universe. For a set system $\cA$, we let $\cA^{k}$ denote the $k$-th Cartesian power of $\cA$ (i.e. $\cA^k=\underbrace{\cA \times \ldots\times\cA}_{k\text{ times}}$)\vspace{-0.9em}, and prove the following:
\begin{lemma}\label{lem:power_efficiency}
    Let $\cA$ be a set system, then $\eta(\cA^{k}) = \eta(\cA)$.
\end{lemma}
\begin{proof}
    By the definition of Cartesian product, we have $|\cA^{k}| = |\cA|^k$. 
    We also note that any maximal chain of $\cA^{k}$ can be described by a tuple formed by the $k$ maximal chains of $k$ copies of $\cA$ and a description of the index of the component from which the next element is taken at each step of the maximal chain.
    Thus, this gives $c(\cA^{k}) = c(\cA)^{k}\cdot \frac{(kn)!}{(n!)^k}$. 
    Then,
    \[\eta(\cA^{k}) = \left(\frac{c(\cA^{k})}{|\cA^{k}|^{2}\cdot (kn)!}\right)^{1/(kn)} = \left(\frac{c(\cA)^{k} \cdot (kn)!}{|\cA|^{2k} \cdot (kn)! \cdot (n!)^k}\right)^{1/(kn)} = \left(\frac{c(\cA)}{|\cA|^{2}\cdot n!}\right)^{1/n} = \eta(\cA).\]
\end{proof}
Furthermore, to obtain a deterministic algorithm, we use the following lemma that is a straightforward application of the probabilistic method:
\begin{lemma}\label{lem:det}
    Let $\cA \subseteq 2^{[n]}$.
	There exists a family $\mathcal{F}$ of $\frac{n!}{c(\cA)}\cdot \poly(n)$ permutations of $[n]$ such that for every permutation $\pi$ of $[n]$ there exists $\pi' \in \mathcal{F}$ such that $\pi' \cdot \pi$ corresponds to a maximal chain of $\cA$, where $(\pi' \cdot \pi)(i) = \pi'(\pi(i))$. Such a family can be found in $\poly(n!)$ time.
\end{lemma}
\begin{proof}
	We use the probabilistic method. If $\pi'$ is a uniformly randomly chosen permutation, then for any permutation $\pi$ the permutation $\pi'\cdot \pi$ is also a uniformly random permutation. Hence, if we construct $\mathcal{F}$ as $\ell:=\frac{n!}{c(\cA)}n \ln n$ uniformly and independently chosen random permutations of $[n]$, then
	\[
	\Pr[\forall \pi' \in \cF : \pi' \cdot \pi \text{ does not correspond to a maximal chain of }\cA] = \left(1 - \frac{c(\cA)}{n!} \right)^\ell \leq \mathrm{e}^{-\ell \frac{c(\cA)}{n!}} \leq n^{-n},
	\]
    where we use the standard inequality $1+x \leq \mathrm{e}^x$ in the first inequality.
	Taking a union bound over the at most $n! < n^{n}$ permutations $\pi$, we see that with probability less than one there exists a $\pi$ such that for all $\pi' \in \cF$ we have that $\pi' \cdot \pi$ does not correspond to a maximal chain of $\cA$.
    In particular, this implies that with positive probability, for every $\pi$ there is a permutation $\pi' \in \cF$ such that $\pi' \cdot \pi$ is represented by a maximal chain of $\cA$.
    Hence, there exists an $\cF$ with this property.
	
    We note that constructing such a family $\cF$ is equivalent to the set cover problem with the universe being $S_n$ and for each $\pi' \in S_n$ a set 
    \[
        S(\pi'):=  \{\pi \in S_n: \pi'\cdot \pi \text{ corresponds to a maximal chain of } \cA\} \subseteq S_n.
    \]
    We already established above that the optimal family of permutations that covers the universe $S_n$ has a size at most $\frac{n!}{c(\cA)}n\ln n$. 
    Thus, a classical greedy algorithm~\cite{Johnson73} with approximation ratio $H(c(\cA)) \le H(n!) \le n \ln n$ where $H(d) = \sum_{i=1}^{d}\frac{1}{i}$ yields a family of permutations of size $\frac{n!}{c(\cA)}\cdot(n\ln n)^2$ and, if naively implemented, runs in time $O((n!)^3\cdot n) = \poly(n!)$.   
\end{proof}
The main purpose of the above lemma is to show that a single set system together with a permutation family of size $\frac{n!}{c(\cA)}\cdot \poly(n)$ is enough to cover all permutations $S_n$. Equipped with this lemma we are able to prove Lemma~\ref{lem:algo}, which we first recall for convenience:

\lemAlgo*
\begin{proof}
    Let $f$, $f_1,\ldots,f_N$ be the function defining the permutation problem that needs to be solved, and let it be of degree $d=O(1)$.
	Let $g$ be a parameter to be set later, let $s = \lceil N/(g\cdot n) \rceil$ and let $N'=g\cdot n\cdot s$. We treat the permutation problem $f$, $f_1,\ldots,f_N$ as a permutation problem over $N'$ elements by defining the function $f_j(\{\sigma_1,\ldots,\sigma_j\},\sigma_{j-d+1},\ldots,\sigma_{j-1},\sigma_j)$ to be $0$ if $\sigma_j\neq j$ for all $j > N$ (hence enforcing that only permutations $\sigma$ satisfying $\sigma_j=j$ for $j > N$ can contribute to the end result).
    
    We divide $[N']$ in groups $V_1,\ldots,V_s$ of size $g\cdot n$ each.	
    For each $i \in [s]$, let $\cA_{i} \subseteq 2^{V_{i}}$ be an arbitrary set system such that $\cA_{i}$ is isomorphic to $\cA^{g}$.\footnote{Set systems $\cA\subseteq 2^U$ and $\cB\subseteq 2^{U'}$ are \emph{isomorphic} if there is a bijection $\varphi: U \rightarrow U'$ such that $A \in \cA$ if and only if $\{\varphi(a):a \in A\} \in\cB$.} 
    For each $i=1,\ldots,s$, apply Lemma~\ref{lem:det} to $\cA_{i}$ to get a permutation family $\cF_i$ of size at most $\frac{(g\cdot n)!}{c(\cA^{g})}\cdot\poly(g\cdot n)$.
	
	Fix $(\pi'_1,\ldots,\pi'_s)\in \cF_1 \times \cdots \times \cF_s$ and define $\cA^{*}$ to be $\cA_{1} \times \cA_{2} \times \cdots \times \cA_{s}$ and $\pi'=P(\pi'_1,\ldots,\pi'_s) \in S_{N'}$ to be the permutation that maps each $a \in V_i$ to $\pi'_i(a)$.
    Letting $C(\cA^*)$ denote the set of permutations corresponding to maximal chains of $\cA^*$ we also define
    \begin{equation}\label{eq:subrout}
        f_{\pi'} := \bigoplus_{\pi'\cdot \sigma \in C(\cA^*)} f(\sigma) = \bigoplus_{\pi'\cdot \sigma \in C(\cA^*)} \bigotimes_{j=1}^{N'} f_j(\{\sigma_1,\ldots,\sigma_j\},\sigma_{j-d+1},\ldots,\sigma_{j-1},\sigma_j).
    \end{equation}
    We have that $\cup_{(\pi'_1,\ldots,\pi'_s) \in \cF_1,\ldots,\cF_s}\{ \sigma \in S_{N'}: P((\pi'_1,\ldots,\pi'_s))\cdot \sigma \in C(\cA^*)\}=S_{N'}$ by the construction of $\cF_1,\ldots,\cF_s$ and $\cA_i$ being isomorphic to $\cA^g$. Since the $\oplus$ operation is idempotent, we therefore have
    \[
    \bigoplus_{\sigma \in S_{N'}} f(\sigma) = \bigoplus_{(\pi'_1,\ldots,\pi'_s)\in \cF_1 \times \cdots \times \cF_s} f_{P(\pi'_1,\ldots,\pi'_s)}.
    \]
    Hence, we can focus on evaluating~\eqref{eq:subrout} for each $\pi'$, since afterwards we can sum all outcomes with a $\left(\frac{(g\cdot n)!}{c(\cA^g)}\cdot \poly(g\cdot n)\right)^s$ multiplicative factor in the running time.

    We do this with a small variant of the relatively standard dynamic programming algorithm (which also occurred in~\cite{KoivistoP10}).
    We will use the notation $\pi'(X)=\{\pi'(e) : e\in X\}$ to check whether a set $X$ is in the set system $\cA$ after having applied the permutation $\pi'$ on it.
    For all $X \subseteq [N]$ such that $\pi'(X) \in \cA$ and $x_{\max\{1,|X|-d+1\}},\ldots,x_{|X|} \in X$ define 
    \[
    \begin{aligned}
        g(X,x_{\max\{1,|X|-d+1\}},\ldots,x_{|X|}) = \bigoplus_{\sigma} \bigotimes_{j=1}^{|X|}f_{j}(\{\sigma_1,\ldots,\sigma_j\},\sigma_{\max\{1,j-d+1\}},\ldots,\sigma_{j-1},\sigma_j),
    \end{aligned}
    \]  
    where the $\oplus$ runs over all permutations $\sigma$ of $X$ with $\sigma_i=x_i$ for all $i=\max\{1,|X|-d+1\},\ldots,|X|$ such that $\pi'(\{\sigma_1,\ldots,\sigma_j\}) \in \cA$ for every $j=1,\ldots,|X|$. By straightforward evaluation, entries $g(X,x_{\max\{1,|X|-d+1\}},\ldots,x_{|X|})$ satisfying $|X| \leq d$ can be evaluated in polynomial time since $d=O(1)$. Moreover, for $|X| > d$, we have the following recurrence that is easy to verify:
    If $\pi'(X \setminus x_{|X|}) \in \cA$, then $g(X,x_{|X|-d},\ldots,x_{|X|})$ is equal to
    \[
        \bigoplus_{x_{|X|-d-1} \in X \setminus \{ x_{|X|-d},\ldots,x_{|X|}\} } g(X \setminus x_{|X|},x_{|X|-d-1},\ldots,x_{|X|-1}) \otimes f_{|X|}(X,x_{|X|-d},\ldots,x_{|X|}),
    \]
    and if $\pi'(X \setminus x_{|X|}) \notin \cA$, then $g(X,x_{|X|-d},\ldots,x_{|X|})=0$. Hence, equipped with this recurrence, we can compute all table entries $g$, and in particular 
    \[
        f_{\pi'} := \bigoplus_{\text{distinct }x_1,\ldots,x_d\in [N']}g([N'],x_1,\ldots,x_d)
    \]
    in time $O^*(|\cA^*|)=O^*(|\cA|^{g\cdot s})$.
    The runtime of invoking Lemma~\ref{lem:det} is $\poly((g\cdot n)!)$ and hence the running time of this algorithm is
	\[
		\poly((g\cdot n)!)+\prod_{i \in [s]} |\cF_i| \cdot |\cA^{*}| \le \poly((g\cdot n)!)\cdot\left(\frac{(g\cdot n)!}{c(\cA^{g})}\cdot\poly(g\cdot n)\right)^s |\cA|^{g \cdot s}.
	\]
	The space usage is $|\cA|^{g \cdot s}+\poly((g\cdot n)!)$. Hence, our space-time product is at most
	\[
	\begin{aligned}
		&\poly((g\cdot n)!)\cdot\left(\frac{(g\cdot n)!}{c(\cA^{g})}\right)^s\cdot |\cA|^{2\cdot g \cdot s}\cdot (g\cdot n)^{O(s)}\\
        \leq&\left(\frac{|\cA^{g}|^2(g\cdot n)!}{c(\cA^{g})}\right)^s \cdot (g\cdot n)^{O(s+g\cdot n)}\\
		=&\eta(\cA^{g})^{-g\cdot s\cdot n}\cdot (g\cdot n)^{O(s+g\cdot n)}\\
		\leq &\eta(\cA)^{-(N+g\cdot n)} \cdot(g\cdot n)^{O(N/(g \cdot n)+g\cdot n)}.
	\end{aligned}
	\]
    For the last inequality, we use Lemma~\ref{lem:power_efficiency} to conclude that $\eta(\cA^g)=\eta(\cA)$, that $N' \leq N + g\cdot n$ and that $s = O(N/ (g\cdot n))$.
    Setting $g = \sqrt{N}/n$, we obtain that the space-time product is at most
    $\eta(\cA)^{-N}N^{O(\sqrt{N})}$.
\end{proof}

\section{Construction of Efficient Posets and Set Systems}
\label{sec:con}
In this section, we show our construction of set systems with high efficiency.
In Subsection~\ref{subsec:impkoiv} we present a simple construction with better efficiency and an analytical analysis; in Subsection~\ref{subsec:best}, we describe our best construction; and in Subsection~\ref{subsec:counter}, we observe that the bucket order construction from~\cite{KoivistoP10} was, coincidentally, conjectured by~\cite{JohnsonLR15} to be optimal from the perspective of set systems.
Thus we use our construction to disprove the conjecture of~\cite{JohnsonLR15}.

\subsection{A Simple Improvement over~\cite{KoivistoP10}\label{subsec:impkoiv}}
We start with the optimal construction from~\cite{KoivistoP10}, the complete bipartite poset (``bucket order'') with $13$ vertices on both sides, which has efficiency $1/3.9271$ and show that by removing a perfect matching from it, the efficiency of the poset is further improved.

\begin{lemma}\label{lem:impkoiv}
    Let $P_{n}$ be a bipartite poset with parts $X =\{x_{0},\ldots,x_{n-1}\}$ and $Y = \{y_0,\ldots,y_{n-1}\}$ such that $(x_{i}, y_{j}) \in P_{n}$ if and only if $i\neq j$. Then $\alpha(P_{n}) = 2^{n+1}+n-1$ and $\lambda(P_{n}) = (n-1)!\cdot n! \cdot (n+1)$. In particular, $\eta(P_{13}) > 1/3.9162.$
\end{lemma}
\begin{proof}
    To count the number of ideals, we consider cases where we take 0, 1 or at least 2 vertices from $Y$ and count their contribution respectively.
    We obtain
    \[\alpha(P_{n}) = \binom{n}{0}2^{n} + \binom{n}{1}2^{1} + \sum_{i=2}^{n}\binom{n}{i} = 2^{n+1}+n-1.\]
    For the possible configurations for linear extensions of $P_{n}$, either all elements of $X$ lie before any element of $Y$ or there is a single element $x_{i}$ that lies on the $(n+1)$-th position after $y_{i}$. 
    We obtain
    \[\lambda(P_{n}) = (n!)^{2} + n\cdot ((n-1)!)^2 = \left(n! + (n-1)!\right)\cdot n! = (n-1)!\cdot n! \cdot (n+1).\]
    Taking $n=13$, we have that $\eta(P_{13})$ equals
    \[\left(\frac{12!\cdot 13! \cdot 14}{(2^{14}+13-1)^2\cdot 26!}\right)^{1/26} \approx 1/3.9161.\]
\end{proof}

\subsection{Our Best Currently Verifiable Construction for Posets}\label{subsec:best}
Motivated by the efficiency improvement from $d$-regular bipartite posets where $d < n$,
we propose an improved construction of a $d$-regular bipartite poset for a small constant $d$ with an efficiency of $1/3.7492$ as verified by a computer program.
\lemConstruction*
\begin{proof}
Let $P$ be a $6$-regular bipartite poset defined as above.
This poset has $\alpha(P) = 2125130762$ and 
\[\lambda(P) = 5463391192321648360195359004759601753062414786866369527808000000.\]
Together, we have $\eta(P) = (\lambda(P)/(\alpha(P)^2\cdot 58!))^{1/58} \approx 1/3.7492$.
In the Appendix~\ref{sec:AppendixImplementation}, we explain how we compute the number of ideals and linear extensions of a $d$-regular bipartite poset and provide a link to the code used for the computation.
\end{proof}

\subsection{A Counterexample to the Conjecture of~\cite{JohnsonLR15}}\label{subsec:counter}
Building on the efficiency improvement from various $d$-regular bipartite posets,
we disprove the conjecture of~\cite{JohnsonLR15}.
Specifically, the authors of~\cite{JohnsonLR15} proposed the following notion:
\begin{definition}[Tower of $t$-cubes~\cite{JohnsonLR15}]
Let $X$ be an $n$-element set, and suppose that $n=tk$ for some integers $t,k \ge 1$.
Partition $X$ into pairwise disjoint blocks
\[
X_1,\dots,X_k
\]
with $|X_i|=t$ for all $i \in [k]$.
The \emph{tower of $t$-cubes} is the family
\[
\mathcal{T}_t
=
\Bigl\{
A \subseteq X :
(X_1 \cup \cdots \cup X_s) \subseteq A \subseteq (X_1 \cup \cdots \cup X_{s+1})
\text{ for some } 0 \le s \le k-1
\Bigr\}.
\]
\end{definition}

We note that this definition coincides with bucket orders from the perspective of set systems.
For example, consider $t = n/2$ for $\cT_{n/2}$ and a two-level bucket order with $n/2$ elements on each level, then a set $A$ is in $\cT_{n/2}$ if and only if it corresponds to an ideal of the bucket order, and a sequence $A_{0},A_{1},\ldots,A_{n} \in \cT_{n/2}$ is a maximal chain if and only if it is represented by a linear extension of the bucket order poset.

Johnson, Leader, and Russell~\cite{JohnsonLR15} conjectured that these constructions are extremal for the number of maximal chains in the Boolean lattice. More precisely, they stated the following:

\begin{conjecture}[Johnson, Leader, and Russell~\cite{JohnsonLR15}, Conjecture 5]
If $|\mathcal{A}| = |\mathcal{T}_t|$, then
\[
c(\mathcal{A}) \le c(\mathcal{T}_t).
\]
\end{conjecture}

In contrast to the above conjecture, we prove that it does not hold in general.
Indeed, we provide a counterexample showing that the tower-of-cubes construction is not always extremal for the number of maximal chains.
Our previous improved construction does have better efficiency, but it does not directly answer the conjecture since given the same size of universe, the size of the set system from our construction is larger than $\cT_{t}$. Thus we add some dummy vertices and arcs in our poset construction to address this issue.

\begin{theorem}
    There exists a set system $\cA$ where $|\cA| < |\cT_{n/2}|$ while $c(\cA) > c(\cT_{n/2})$.
\end{theorem}
\begin{proof}
    We give a specific set system based on a modified $7$-regular bipartite poset. 
    
    Let $V = \{x_{0},y_{0},x_{1},y_{1},\ldots,x_{15},y_{15}, d_{0}, d_{1}\}$ where $|V| = 34$, 
    we have a poset $P$ as
    \begin{itemize}
        \item base part: $(x_{i}, y_{j}) \in P$ if and only if $(i-j) \mod 16 \in \{0, 1, \ldots, 6\}$;
        \item dummy elements and arcs: $(x_{0},y_{8})$, $(y_{8},d_{0})$, $(y_{15}, d_{1}) \in P$.
    \end{itemize}
    We have $\alpha(P) = 260553$ and \[\lambda(P) = 131576429145341435860520294400\] as verified by a counting program for general posets, which is also described in the appendix~\ref{sec:AppendixIrregularPoset}.
    Then by Observation~\ref{obs:posetEfficiency}, 
    we define \[\cA := \{X \subseteq V \mid X \text{ is an ideal of } P\},\]
    and we note that $|\cA| = \alpha(P)$ and $c(\cA) = \lambda(P)$.
    Finally, consider $\cT_{n/2}$ with $n = 34$, it gives
    \[|\cT_{n/2}| = 2^{18} - 1 = 262143 > |\cA|, \text{ and }\frac{c(\cT_{n/2})}{c(\cA)}  = \frac{(17!)^2}{\lambda(P)} \approx 0.96 < 1.\]
\end{proof}
Using a similar technique, we can also obtain a construction that serves as a counterexample to Conjecture~6 of~\cite{JohnsonLR15} for the ``generalized tower-of-cubes''.

\subsection{Our Best Set System Construction}
We show a simple construction of set systems with better efficiency where the size of the set system and the number of maximal chains contained in it can be computed easily.

Throughout this and subsequent sections, we denote \( h(p) := -p\log_2 p - (1-p)\log_2(1-p) \) with the convention that $h(0) = h(1) = 0$ for the binary entropy function. It is well known that
\[
\binom{n}{pn} = 2^{h(p)n + O(\log n)}.
\]
For \( x \in (0,1) \), we define \( h^{-1}(x) \) as the inverse of \( h \) restricted to \( [0,1/2] \).
Let $n=2m$, and fix an arbitrary bipartition of $[n]$ as $L \sqcup R = [n]$ with $|L| = |R| = m$. 
For a parameter $\tau \in [1,2]$, define
\[
\cA_{m,\tau} := \left\{S \subseteq [n] \middle\vert\ h\left(\frac{|L\cap S|}{m}\right) + h\left(\frac{|R\cap S|}{m}\right) \le \tau \right\}.
\]
Thus membership in $\mathcal{A}_{m,\tau}$ only depends on the number of chosen elements from the two parts $L$ and $R$, and we have 
\[|\cA_{m,\tau}| = \sum_{i=0}^{m}\sum_{j=0}^{m} \mathbf{1}_{h(i/m) + h(j/m) \le \tau} \binom{m}{i} \binom{m}{j},\]
where $\mathbf{1}_{\{\cdot\}}$ is the indicator function.
\begin{definition}
For $0\leq i,j \leq m$, a sequence $\sigma \in \{\mathbb{L},\mathbb{R}\}^{i+j}$ is called an \emph{$(i,j)$-sequence} if it contains exactly $i$ occurrences of $\mathbb{L}$ and exactly $j$ occurrences of $\mathbb{R}$. We denote it by
\[|\sigma|_{\mathbb{L}} = i \text{ and } |\sigma|_{\mathbb{R}} = j.\]
Furthermore, an $(i,j)$-sequence $\sigma$ is \emph{$\tau$-admissible} 
if for every prefix $\omega$ of $\sigma$, it holds that
    \begin{equation}\label{eq:contr}
        h\left(\frac{|\omega|_{\mathbb{L}}}{m}\right) + h\left(\frac{|\omega|_{\mathbb{R}}}{m}\right) \leq \tau.
    \end{equation}
\end{definition}

\begin{lemma}
For every $m \ge 1$ and $\tau \in [1, 2]$,
\[c(\cA_{m,\tau}) = (m!)^2 \cdot |\{\sigma : \sigma \text{ is a $\tau$-admissible $(m,m)$-sequence}\}|.\]
\end{lemma}
\begin{proof}
    Let $\mathcal{S}$ denote the set of all $\tau$-admissible $(m,m)$-sequences and $S_{m}$ denote the set of all permutations of $m$ elements.
    We show that there exists a bijection $\phi$ between the maximal chains supported by $\cA_{m,\tau}$ and $\mathcal{S} \times (S_{m})^2$. We will use the following notation: A \emph{prefix subset} of a maximal chain $\emptyset \subset \{x_1\} \subset \{x_1,x_2\} \subset  \cdots \subset [2m]$
    is a set $\{x_1,\ldots,x_i\}$, for some $i=0,\ldots,2m$.
    
    Take an arbitrary maximal chain supported by $\cA_{m,\tau}$ as
    \[
    \emptyset \subset \{x_1\} \subset \{x_1,x_2\} \subset  \cdots \subset [2m].
    \]
    We define $\phi(x_1,\ldots,x_{2m})=(\sigma,\pi_\mathbb{L},\pi_\mathbb{R})$ where $\sigma$ is the $(m,m)$-sequence where the $t$-th symbol is $\mathbb{L}$ if $x_{t} \in L$, and is $\mathbb{R}$ if $x_{t} \in R$.
    Since all prefix subsets of the maximal chain are  in $\cA_{m,\tau}$, the corresponding $(m,m)$-sequence is $\tau$-admissible by definition of $\cA_{m,\tau}$.

    We define $\pi_{\mathbb{L}}$ to be the order in which the elements of $L$ appear in $(x_1,\ldots,x_{2m})$, and define $\pi_{\mathbb{R}}$ to be the order in which the elements of $R$ appear in $(x_1,\ldots,x_{2m})$. 
    Furthermore, it is clear that $(\sigma,\pi_{\mathbb{L}},\pi_{\mathbb{R}})$ uniquely determines $(x_1,\ldots,x_{2m})$.

    Conversely, given $(\sigma,\pi_{\mathbb{L}},\pi_{\mathbb{R}}) \in \mathcal{S} \times (S_{m})^2$, we construct a maximal chain as follows.
    We read $\sigma$ from left to right, and whenever the next symbol is $\mathbb{L}$, insert the next unused element in the ordering $\pi_{\mathbb{L}}$.
    A similar process applies to elements in $\pi_{\mathbb{R}}$.
    Take any prefix subset $c'$ of the constructed maximal chain $c$, similarly, it corresponds to an $(i,j)$-sequence by counting the occurrences of elements from $L$ or $R$. 
    If the prefix subset is not included by $\cA_{m,\tau}$, then the $(i,j)$-sequence cannot be $\tau$-admissible and thus the $(m,m)$-sequence is not $\tau$-admissible as well.
    The contradiction indicates that $c$ is supported by $\cA_{m,\tau}$.
    Furthermore if $(\sigma,\pi_{\mathbb{L}},\pi_{\mathbb{R}})\neq (\sigma',\pi'_{\mathbb{L}},\pi'_{\mathbb{R}})$ then applying the process to $(\sigma,\pi_{\mathbb{L}},\pi_{\mathbb{R}})$ clearly results in a different maximal chain then after applying the process to $(\sigma',\pi'_{\mathbb{L}},\pi'_{\mathbb{R}})$.
\end{proof}

The number of $\tau$-admissible $(m,m)$-sequences $S[m,m]$ is computed by the following recurrence.
\begin{lemma}
For $0 \leq i,j \leq m$, let $S[i,j]$ be the number of $\tau$-admissible $(i,j)$-sequences. Then we have
\[
    S[i,j]   =
    \begin{cases}
        1, & \text{if $i=0$ or $j=0$},\\
        S[i-1,j] + S[i,j-1], & \text{if } i, j \geq 1, \text{ and } h(i/m) + h(j/m) \le \tau,\\
        0, &  \text{if } i, j \geq 1, \text{ and } h(i/m) + h(j/m) > \tau.
    \end{cases}
\]
\end{lemma}
\begin{proof}
If $i=0$, then the only sequence can be $\mathbb{R}\cdots \mathbb{R}$ (the sequence with $j$ copies of symbol $\mathbb{R}$ ). This sequence is a $\tau$-admissible $(0,j)$-sequence since $\tau \geq 1$ and the first term of~\eqref{eq:contr} is always equal to $0$. 
The case $j=0$ is symmetric.

Now assume that $i,j \ge 1$, if $h(i/m)+h(j/m) > \tau$, then no $(i,j)$-sequence can be $\tau$-admissible, since the sequence itself is also one of the prefix sequences.

Otherwise, if $h(i/m)+h(j/m)\leq \tau$, 
then we note every $\tau$-admissible $(i,j)$-sequence ends either with symbol $\mathbb{L}$ or with symbol $\mathbb{R}$. Deleting the last symbol gives, respectively, a $\tau$-admissible $(i-1,j)$-sequence or a $\tau$-admissible $(i,j-1)$-sequence.
This gives the recurrence.
\end{proof}
Thus, for any fixed choice of $m$ and $\tau$, the efficiency of $\cA_{m,\tau}$ can be verified by evaluating the above dynamic program together with the explicit formula for $|\cA_{m,\tau}|$.

In the Appendix~\ref{sec:AppendixImplementation}, we provide a link to the code used for computation. It turns out that with the choice of $m=50000$ and $\tau = 1.032$, we have $\eta(\cA_{m,\tau}) \approx 1/3.1860$, which proves Lemma~\ref{lem:setSystemConstruction}.

\section{Limits of our Framework: Upper Bounds on Chain Efficiency}
\label{sec:upp}

In this section, we derive upper bounds on the chain efficiency of set families. In Subsection~\ref{Sec:BasicUpperBound}, we present a basic upper bound, which we subsequently refine in Subsection~\ref{sec:improvedBound} using more advanced techniques.
Finally, in Subsection~\ref{sec:ubBiPoset}, we give an improved upper bound for regular bipartite posets, which gives a stronger bound on the efficiency of interesting families of posets, including bucket orders and those described in Subsection~\ref{sec:con}. 
This indicates that our construction in Subsection~\ref{subsec:best} is quite competitive within the family of regular bipartite posets.

\subsection{Warm-up: Basic Upper Bound}~\label{Sec:BasicUpperBound}
We begin with an initial upper bound on chain efficiency.
\begin{lemma}\label{lem:basicUpperbound}
    Let $n\ge3$ and let $\cT \subseteq 2^{[n]}$. The efficiency of $\cT$ is at most $\frac{1+\frac{1}{\poly(n)}}{3}$.
\end{lemma}
\begin{proof}
 For convenience, assume that $n$ is divisible by three. With any maximal chain $P:=A_0,\dots,A_n$ of $\mathcal{T}$ we define $L(P):= A_{\frac{n}{3}}$  and  $R(P):= A_{\frac{2n}{3}}$. 

Note that $P$ can be described by $L(P),R(P)$ and three permutations $\pi_1$, $\pi_2$ and $\pi_3$ on $n/3$ elements, where $\pi_1$ is a permutation of $L(P)$, $\pi_2$ is a permutation of $R(P)\setminus L(P)$, and $\pi_3$ is a permutation of $[n]\setminus R(P)$. Hence it holds that:

\begin{equation}\label{eq:lb}
    c(\mathcal{T}) \leq |\mathcal T|^2 ((n/3)!)^3.
\end{equation}

Note that~\eqref{eq:lb} immediately implies that
\[
    \frac{1}{\eta(\cT)} = \left( \frac{|\cT|^2 \cdot  n!}{c(\cT)} \right)^{1/n} \geq \left( \frac{|\mathcal{T}|^2 \cdot  n!}{|\mathcal{T}|^2 ((n/3)!)^3} \right)^{1/n} \geq \left(\frac{3^n}{(\mathbf{e}\cdot n)^3}\right)^{1/n} \geq \frac{3}{(\mathbf{e}\cdot n)^{3/n}},
\]
where we used in the second inequality $(\frac{n}{\mathbf{e}})^n\leq n! \leq \mathbf{e}\cdot n (\frac{n}{\mathbf{e}})^n$ to conclude that
\[
\frac{n!}{((n/3)!)^3}\geq \frac{(n/\mathbf{e})^n}{(\mathbf{e}\cdot n)^3(n/(3\mathbf{e}))^{n}} \geq \frac{3^n}{(\mathbf{e}\cdot n)^3}.
\]
Finally, to complete the proof we need to show that,
$$(\mathbf{e}\cdot n)^{3/n}\le 1+\frac{6\log_2 n}{n}.$$
To see this, observe that

$$\left(1+\frac{6\log_2 n}{n}\right)^n=\left(\left(1+\frac{6\log_2 n}{n}\right)^{n/6 \log_2 n}\right)^{6\log_2 n}\ge (2^{6})^{\log_2 n}=n^{6}\ge (n\cdot \mathbf{e})^3,$$
where we used the fact that $(1 + 1/x)^x \ge 2$, for every positive integer $x\ge 1$.
\end{proof}
\subsection{An Improved Upper Bound on the Efficiency of Set Systems}\label{sec:improvedBound}

To present our improved upper bound, we introduce several notions and results from coding theory, including the Hamming distance and extremal properties of Hamming balls.

\begin{definition}[Hamming distance, Hamming ball]
Let $X$ be a finite set. For any two subsets $A,B \subseteq X$, the \emph{Hamming distance} between $A$ and $B$ is defined as
\[
d(A,B) = |A \triangle B|,
\]
where $A \triangle B = (A \setminus B) \cup (B \setminus A)$ denotes the symmetric difference.

A \emph{Hamming ball centered at a set $X$} is a set system $\cA$ with the property that, if $d(A,X) < d(A',X)$ and $A' \in \cA$, then also $A \in \cA$.

\end{definition}

\begin{definition}[Distance between set systems]
Let $\mathcal{A}, \mathcal{B} \subseteq \mathcal{P}(X)$ be two nonempty set systems. The distance between $\mathcal{A}$ and $\mathcal{B}$ is defined as
\[
d(\mathcal{A}, \mathcal{B}) = \min \{ d(A,B) : A \in \mathcal{A},\ B \in \mathcal{B} \}.
\]
\end{definition}

The following theorem, a form of Harper’s theorem~\cite{Harper1966} due to Frankl and Füredi, shows that Hamming balls are extremal with respect to distances between set systems.
\begin{theorem}[Frankl and Füredi~\cite{FranklFuredi1981}]\label{thm:ff}
Let $\mathcal{A}, \mathcal{B} \subseteq \mathcal{P}(X)$ be nonempty set systems. Then there exist Hamming balls $\mathcal{A}_0$ centered at $X$ and $\mathcal{B}_0$ centered at $\emptyset$ such that
\[
|\mathcal{A}_0| = |\mathcal{A}|, \qquad |\mathcal{B}_0| = |\mathcal{B}|,
\]
and
\[
d(\mathcal{A}_0, \mathcal{B}_0) \ge d(\mathcal{A}, \mathcal{B}).
\]
\end{theorem}

In particular, among all pairs of set systems of given sizes, Hamming balls maximize the minimum distance.

\begin{corollary}\label{cor:DistanceUpperBound}
Let $\cX,\cY \subseteq 2^{U}$ be non-empty set families of size at least $2^{\lambda n}$, then $d(\cX_0,\cY_0)\le(1-2h^{-1}(\lambda))|U|$.
\end{corollary}

\begin{proof}
    
We apply Theorem~\ref{thm:ff} to obtain $\cX_0$ and $\cY_0$ as stated in that theorem. If $|\cX|, |\cY|\ge2^{\lambda n}$, then $\cY_0$ contains all sets of size $h^{-1}(\lambda)n$ and $\cX_0$ contains all sets of size $(1-h^{-1}(\lambda))n$ and hence $d(\cX_0,\cY_0)=(1-2\cdot h^{-1}(\lambda))|U|$.\end{proof}

\begin{lemma}\label{lem:cons}
Let $\cX,\cY \subseteq 2^{U}$ be non-empty set families with $|\cX|,|\cY|\geq 2^{h(\delta) |U|+1}$, then there exists a subset $\cX'\subseteq \cX$ of size $|\cX|/2$ such that for every element $X \in \cX'$ we have  \[d(X,\cY) \leq (1-2\cdot \delta)|U|.\]
\end{lemma}

\begin{proof}
Start with the families $\cX$ and $\cY$, and initialize $\cX' \gets \emptyset$. We iteratively construct $\cX'$ as follows.
As long as $|\cX'|<|\cX|/2$, apply Corollary~\ref{cor:DistanceUpperBound} to the current families $\cX$ and $\cY$. This yields a pair $A \in \cX$ and $B \in \cY$ such that
\[
d(A,B)\le (1-2\delta)|U|.
\]
Add $A$ to $\cX'$, and remove it from $\cX$. That is, update
\[
\cX' \gets \cX' \cup \{A\}
\qquad\text{and}\qquad
\cX \gets \cX \setminus \{A\}.
\]

Repeating this procedure until $|\cX'|=|\cX|/2$, we obtain a subfamily $\cX' \subseteq \cX$ of size $|\cX|/2$ such that for every $A \in \cX'$ there exists some $B \in \cY$ satisfying
\[
d(A,B)\le (1-2\delta)|U|.
\]

\end{proof}

Now we have all the ingredients to prove the main result of this section.
\imprb*

\begin{proof}
    As the initial step we upper bound $c(\cA)$ in terms of $|\cA|$. Let
\[
\begin{aligned}
    \cT := \{(W,X,Y,Z) : &\ W,X,Y,Z \text{ are pairwise disjoint,}\\ &|W|=|Z|=n/3, |X|=|Y|=n/6, W, W\cup X, W\cup X\cup Y \in \cA\}.
\end{aligned}
\]
Note that if $A_0,\ldots,A_n$ is a maximal chain of $\cA$, then there exists $(W,X,Y,Z)\in \cT$ where,
\[
W:= A_{n/3}, \quad X:= A_{n/2} \setminus A_{n/3}, \quad Y:= A_{2n/3} \setminus A_{n/2}, \quad Z:= A_{n} \setminus A_{2n/3}.
\]
Therefore, we have that 
\begin{equation}
c(\cA)\leq |\cT|\cdot (|W|!)(|X|!)(|Y|!)(|Z|!)= |\cT| \cdot ((n/3)!)^2((n/6)!)^2.
\end{equation}
Let $0 < \delta < 1$ be a parameter to be chosen later, and define $(W,X,Y,Z) \in \cT$ to be \emph{frequent} if
\[
    |\{ (X',Y'):  (W,X',Y',Z) \in \cT \}|\geq 2^{h(\delta) n/3+1}.
\]
Define $\cT^+$ to be all frequent tuples in $\cT$ and $\cT^-= \cT \setminus \cT^+$, so we have $|\cT| = |\cT^+|+|\cT^-|$. Note that $|\cT^-|\leq |\cA|^22^{h(\delta) n/3+1}$.
We now upper bound $|\cT^+|$. Fix a frequent $(W,X,Y,Z) \in \cT^+$, and let
\[
\begin{aligned}
 \mathcal{X} &:= \{X' : (W,X',Y' ,Z) \in \cT^+ \text{, for some $Y'$}\}\\
 \mathcal{Y} &:= \{Y' : (W,X',Y',Z) \in \cT^+ \text{, for some $X'$}\}.
\end{aligned}
\]

By Lemma~\ref{lem:cons}, there is a set $\cX' \subseteq \cX$ with $|\cX'|\geq |\cX|/2$ such that for each $X \in \cX'$, there is a $\mu(X) \in \cY$ such that $d(X,\mu(X))\leq (1-2\cdot \delta)\frac{n}{3}$.

Now we present the crucial idea. For each $X \in \cX'$, we encode the tuple $(W,X,Y,Z)$ using the following three sets
\[
W \cup X, \qquad \mu(X) \cup Z, \qquad \text{and} \qquad X \Delta \mu(X).
\]
We show that this encoding uniquely determines $(W,X,Y,Z)$:

\begin{itemize}
\item Observe that both $X$ and $\mu(X)$ are disjoint from $W$. Hence,
\[
(W \cup X) \setminus (X \Delta \mu(X)) = W \cup (X \cap \mu(X)).
\]
\item Similarly, since both $X$ and $\mu(X)$ are disjoint from $Z$, we obtain
\[
(\mu(X) \cup Z) \setminus (X \Delta \mu(X)) = Z \cup (X \cap \mu(X)).
\]
\item Since $W$, $X$, and $Z$ are pairwise disjoint, the sets $W$ and $Z$ can now be recovered as
\[
W = \bigl(W \cup (X \cap \mu(X))\bigr) \setminus \bigl(Z \cup (X \cap \mu(X))\bigr),
\]
\[
Z = \bigl(Z \cup (X \cap \mu(X))\bigr) \setminus \bigl(W \cup (X \cap \mu(X))\bigr).
\]
\item Once $W$ is known, we recover $X$ from $W \cup X$ using the fact that $W$ and $X$ are disjoint:
\[
X = (W \cup X) \setminus W.
\]
\item Finally, $Y$ is uniquely determined by
\[
Y = [n] \setminus (W \cup X \cup Z).
\]
\end{itemize}

Note that $W\cup X \in \cA$ and $[n] \setminus (\mu(X) \cup Z) \in \cA$ (since $(W,X',\mu(X),Z) \in \cT$ for some $X'$). Hence, the number of options for the three sets is $|\cA|^2\binom{n}{(1-2\delta)n/3}$. Thus,
\[
|\cT| = |\cT^-| + |\cT^+| \leq |\cA|^2 \left(2^{h(\delta)n/3+1}+\binom{n}{(1-2\delta)n/3}\right) \leq |\cA|^2 \left(2^{h(\delta)n/3+1} + 2^{h((1-2\delta)/3)n}\right).  
\]
Now we determine the value of $\delta$ that minimizes the expression
\[
    \gamma:= \max\left\{h(\delta) /3, h((1-2\delta)/3)\right\}.
\]
This happens at $\delta\approx 0.41069$ which yields $\gamma \leq 0.326$. Therefore:
\[
\begin{aligned}
    \eta(\cA) &= \left(\frac{c(\cA)}{|\cA|^2\cdot n!}\right)^{1/n}\\
    &\leq \left(\frac{|\cT| ((n/3)!)^2((n/6)!)^2}{|\cA|^2\cdot n!}\right)^{1/n}\\
    &=(|\cT|\cdot|\cA|^{-2})^{1/n}\cdot (1/3)^{2/3} (1/6)^{1/3}\\
    &= 2^\gamma \cdot (1/3)^{2/3} (1/6)^{1/3}\\
    &\leq 0.331643 \leq 1/3.015.
\end{aligned}
\]
\end{proof}
\subsection{Upper Bound on Bipartite Posets}\label{sec:ubBiPoset}
We begin this section by presenting a lower bound on the number of ideals of a bipartite poset.
\begin{lemma}\label{lem:LowerboundIdealsBipartiteRegular}
    Given a \(d\)-regular bipartite poset \(P\) on $2n$ elements,
    \[ \alpha(P)\ge \sum_{i=0}^n \binom{n}{i}2^{n\frac{\binom{n-d}{i}}{\binom{n}{i}}}.\]
\end{lemma}
\begin{proof}
    Assume the bi-partition is \((X,Y)\), such that for every arc \((x,y)\) it holds that \(x\in X\) and \(y \in Y\). Note that since \(P\) is \(d\)-regular, we have \(|X|=|Y|=n\). Note that given a subset \(X'\) of \(X \) there are exactly \(2^{|Y\setminus N(X')|}\) ideals \(I\) of \(P\) such that \(X\setminus I=X'\).

    We first claim that \[\sum_{X'\subseteq X,|X'|=i}|Y\setminus N(X')|=n\binom{n-d}{i}.\] This holds as for a fixed \(y\in Y\) there are \(\binom{n-d}{i}\) subsets \(X' \subseteq X\) such that \(|X'|=i\) and \(y\notin N(X')\). Furthermore, 
    \[\alpha(P)=\sum_{X'\subseteq X}2^{n-|N(X')|}.\]
    As \(2^x\) is a convex function, then for a fixed \(i\):
    \[\sum_{X'\subseteq X,|X'|=i}2^{n-|N(X')|}\ge\binom{n}{i} 2^{n\frac{\binom{n-d}{i}}{\binom{n}{i}}},\]
    and hence the claim.
\end{proof}

Now we can show the main result of this subsection, which we first recall for convenience:
\bipreg*

\begin{proof}
    Let \(P\) be a \(d\)-regular bipartite poset on \(n\) elements.
    Authors in~\cite{Brightwell2023} showed that:
    \[\lambda(P) \le n!\binom{2d}{d}^{-n/2d}.\]
    Combining this with our lower bound on the number of ideals presented by Lemma~\ref{lem:LowerboundIdealsBipartiteRegular}, we get:
    \[\eta(P)=\left(\frac{\lambda(P)}{\alpha(P)^2 n!}\right)^{1/n}\le \left(\frac{n!\binom{2d}{d}^{-n/2d}}{\alpha(P)^2  n!}\right)^{1/n}\le \frac{\binom{2d}{d}^{-1/2d}}{\left(\sum_{i=0}^{n/2} \binom{n/2}{i}2^{n/2\cdot\frac{\binom{n/2-d}{i}}{\binom{n/2}{i}}}\right)^{2/n}}\]
   Observe that, in order to prove the lemma, we can assume \(n\) to be larger than \(n_0\), for any positive integer \(n_0\). This is based on the fact that by Lemma~\ref{lem:power_efficiency} (which also holds for posets) for every positive integer \(c\), \(P^c\) is also a \(d\)-regular poset on \(nc\) elements with \(\eta(P^c)=\eta(P)\).
    
    In the following claim we describe our desired lower bound for \(n\) which is sufficient to complete our argument (The proof of the claim appears in Appendix~\ref{sec:MissingProofs}).
    \begin{claim}\label{clm:LBOnn0}
        For any positive integer \(d\), and any constants \(q\in(0,1)\) and \(\varepsilon\in (0,\frac{1}{2})\), there exists \(n_0\) such that if \(n>n_0\), then \[\left(\sum_{i=0}^{n/2} \binom{n/2}{i}2^{n/2\cdot\frac{\binom{n/2-d}{i}}{\binom{n/2}{i}}}\right)^{2/n}\ge 2^{(h(q)+q^d)(1-\varepsilon)}.\]
    \end{claim}
    So we assume that \(n\) is at least as large as \(n_0\) described in Claim~\ref{clm:LBOnn0}.
    We only need to prove the following inequality
    \[\max_{0<q<1}2^{h(q)+q^d}\cdot{\binom{2d}{d}^{1/2d}}> 3.6.\] The proof appears in Appendix~\ref{sec:MissingProofs}.

\end{proof}

\section{Concluding Remarks}\label{sec:conc}
We improved the best space-time tradeoffs for permutation problems based on a new algorithm that can use any efficient poset (and even more general, set system) for a space and time efficient algorithm, by reducing the design of efficient dynamic programming algorithms to a clean problem in extremal combinatorics.

Note we are still very far from answering the ambitious Open Problem~\ref{op}\textbf{(b)}: Even our most general method via set systems can not directly lead to an $S$ space and $T$ time algorithm for any $S$ and $T$ satisfying $S\cdot T \leq O^{*}(3.015^{N})$ by Lemma~\ref{lem:imprb}. Nevertheless, we believe that our general connection between designing space/time efficient algorithms and questions in extremal combinatorics addressing set systems and their maximal chains may be of use here, simply since it seems to be key in the regime of low exponential space as well.

We believe Open Problem~\ref{opposet} and the question of set systems with optimal efficiency are much more tractable.
A first intuition may suggest that this is a typical isoperimetric problem and that the optimal set system (and its corresponding set of maximal chains) are those that are clustered in some sense. But, unfortunately, the most natural candidates of such set systems (i.e. Hamming balls or (co)-lex prefixes) do not seem very efficient.

As already observed in~\cite{JohnsonLR15} as a first concrete step in this direction, there exist set systems of optimal efficiency that are \emph{left-compressed}. A set system is left-compressed if $i < j$ and $A \in \cA$ such that $i \notin A$ but $j \in A$ then $A \setminus \{j\} \cup \{i\} \in \cA$. As a direction towards further research, we will strengthen this observation here:
Given two permutations $\pi,\pi'$ we say that $\pi \preceq \pi'$ if $\pi$ can be obtained from $\pi'$ by iteratively swapping values $\pi'_i$ and $\pi'_j$ at locations $i<j$ such that $\pi'_j < \pi'_i$. Since $\cA$ is left-compressed it follows that, if $\pi \preceq \pi'$ and $\pi'$ is a maximal chain of $\cA$, then so is $\pi$.

With $A = \{a_1,\ldots,a_\ell\} \subseteq [n]$ with $a_1 < \ldots < a_\ell$ and $[n] \setminus A = \{b_1,\ldots,b_{n-\ell}\}$ with $b_1 < \ldots < b_{n-\ell}$ we associate the permutation $\pi(A) = (a_1,\cdots,a_\ell,b_1,\cdots,b_{n-\ell})$. Intuitively, $\pi(A)$ is the unique minimal (under the relation $\preceq$) permutation that has a prefix with elements equal to $A$. Then the observation is that, since in a set system $\cA$ of optimal efficiency every set occurs in some maximal chain, there exists a set system $\cA^*$ of optimal efficiency such that if $\pi(A) \preceq \pi'(A')$ and $A \in \cA^*$ then $A'\in \cA^*$.
Note that this is a stronger property than $\cA^*$ being left-compressed (for example $\{4\} \in \cA^*$ now implies $\{1,2,4\} \in \cA^*$ since $1243 \preceq 4123$).

\section{AI Disclosure}
We used ChatGPT 5.4, ChatGPT 5.5, and Gemini Pro to assist with implementing code for our algorithms to compute the chain efficiencies of some specific posets. The tool materially affected Section~\ref{sec:con}, and a link to all the implemented codes and also additional details to understand the algorithms used in the codes are provided in Appendix~\ref{sec:AppendixImplementation}. All content, including the implemented code and references, has been reviewed and verified by the authors for accuracy and originality.

\bibliographystyle{abbrv}
\bibliography{ref}

@inproceedings{KoivistoP10,
	author       = {Mikko Koivisto and
		Pekka Parviainen},
	editor       = {Moses Charikar},
	title        = {A Space-Time Tradeoff for Permutation Problems},
	booktitle    = {Proceedings of the Twenty-First Annual {ACM-SIAM} Symposium on Discrete
		Algorithms, {SODA} 2010, Austin, Texas, USA, January 17-19, 2010},
	pages        = {484--492},
	publisher    = {{SIAM}},
	year         = {2010},
	url          = {https://doi.org/10.1137/1.9781611973075.41},
	doi          = {10.1137/1.9781611973075.41},
	bibsource    = {dblp computer science bibliography, https://dblp.org}
}

@article{Brightwell2023,
  author       = {Graham R. Brightwell and
                  Prasad Tetali},
  title        = {The Number of Linear Extensions of the {B}oolean Lattice},
  journal      = {Order},
  volume       = {20},
  number       = {4},
  pages        = {333--345},
  year         = {2003},
  url          = {https://doi.org/10.1023/B:ORDE.0000034596.50352.f7},
  doi          = {10.1023/B:ORDE.0000034596.50352.F7},
  bibsource    = {dblp computer science bibliography, https://dblp.org}
}

@article{Harper1966,
  author  = {Harper, L. H.},
  title   = {Optimal numberings and isoperimetric problems on graphs},
  journal = {Journal of Combinatorial Theory},
  volume  = {1},
  year    = {1966},
  pages   = {385--393}
}

@article{FranklFuredi1981,
  author  = {Frankl, P. and F{\"u}redi, Z.},
  title   = {A short proof for a theorem of {H}arper about {H}amming-spheres},
  journal = {Discrete Mathematics},
  volume  = {34},
  year    = {1981},
  pages   = {311--313}
}

@article{hunter2024disjoint,
  title={Disjoint pairs in set systems and combinatorics of low rank matrices},
  author={Hunter, Zach and Milojevi{\'c}, Aleksa and Sudakov, Benny and Tomon, Istv{\'a}n},
  journal={arXiv preprint arXiv:2411.13510},
  year={2024}
}

@article{JohnsonLR15,
  author       = {J. Robert Johnson and
                  Imre Leader and
                  Paul A. Russell},
  title        = {Set Systems Containing Many Maximal Chains},
  journal      = {Comb. Probab. Comput.},
  volume       = {24},
  number       = {3},
  pages        = {480--485},
  year         = {2015},
  doi          = {10.1017/S0963548314000510},
}

@inproceedings{kahn1992entropy,
  title={Entropy and sorting},
  author={Kahn, Jeff and Kim, Jeong Han},
  booktitle={Proceedings of the twenty-fourth annual ACM symposium on Theory of computing},
  pages={178--187},
  year={1992}
}

@article{DittmerP20,
  author       = {Samuel J. Dittmer and
                  Igor Pak},
  title        = {Counting Linear Extensions of Restricted Posets},
  journal      = {Electron. J. Comb.},
  volume       = {27},
  number       = {4},
  pages        = {4},
  year         = {2020},
  url          = {https://doi.org/10.37236/8552},
  doi          = {10.37236/8552},
  bibsource    = {dblp computer science bibliography, https://dblp.org}
}

@inproceedings{KangasHNK16,
  author       = {Kustaa Kangas and
                  Teemu Hankala and
                  Teppo Mikael Niinim{\"{a}}ki and
                  Mikko Koivisto},
  editor       = {Subbarao Kambhampati},
  title        = {Counting Linear Extensions of Sparse Posets},
  booktitle    = {Proceedings of the Twenty-Fifth International Joint Conference on
                  Artificial Intelligence, {IJCAI} 2016, New York, NY, USA, 9-15 July
                  2016},
  pages        = {603--609},
  publisher    = {{IJCAI/AAAI} Press},
  year         = {2016},
  url          = {http://www.ijcai.org/Abstract/16/092},
  bibsource    = {dblp computer science bibliography, https://dblp.org}
}

@article{GurevichS87,
  author       = {Yuri Gurevich and
                  Saharon Shelah},
  title        = {Expected Computation Time for {H}amiltonian Path Problem},
  journal      = {{SIAM} J. Comput.},
  volume       = {16},
  number       = {3},
  pages        = {486--502},
  year         = {1987},
  url          = {https://doi.org/10.1137/0216034},
  doi          = {10.1137/0216034},
  bibsource    = {dblp computer science bibliography, https://dblp.org}
}

@book{FominK10,
  author       = {Fedor V. Fomin and
                  Dieter Kratsch},
  title        = {Exact Exponential Algorithms},
  series       = {Texts in Theoretical Computer Science. An {EATCS} Series},
  publisher    = {Springer},
  year         = {2010},
  url          = {https://doi.org/10.1007/978-3-642-16533-7},
  doi          = {10.1007/978-3-642-16533-7},
  isbn         = {978-3-642-16532-0},
  bibsource    = {dblp computer science bibliography, https://dblp.org}
}

@inproceedings{NederlofW21,
  author       = {Jesper Nederlof and
                  Karol Wegrzycki},
  editor       = {Samir Khuller and
                  Virginia Vassilevska Williams},
  title        = {Improving {S}chroeppel and {S}hamir's algorithm for subset sum via orthogonal vectors},
  booktitle    = {{STOC} '21: 53rd Annual {ACM} {SIGACT} Symposium on Theory of Computing,
                  Virtual Event, Italy, June 21-25, 2021},
  pages        = {1670--1683},
  publisher    = {{ACM}},
  year         = {2021},
  url          = {https://doi.org/10.1145/3406325.3451024},
  doi          = {10.1145/3406325.3451024},
  bibsource    = {dblp computer science bibliography, https://dblp.org}
}

@inproceedings{AustrinKKM13,
  author       = {Per Austrin and
                  Petteri Kaski and
                  Mikko Koivisto and
                  Jussi M{\"{a}}{\"{a}}tt{\"{a}}},
  editor       = {Fedor V. Fomin and
                  Rusins Freivalds and
                  Marta Z. Kwiatkowska and
                  David Peleg},
  title        = {Space-Time Tradeoffs for {S}ubset {S}um: An Improved Worst Case Algorithm},
  booktitle    = {Automata, Languages, and Programming - 40th International Colloquium,
                  {ICALP} 2013, Riga, Latvia, July 8-12, 2013, Proceedings, Part {I}},
  series       = {Lecture Notes in Computer Science},
  pages        = {45--56},
  publisher    = {Springer},
  year         = {2013},
  url          = {https://doi.org/10.1007/978-3-642-39206-1\_5},
  doi          = {10.1007/978-3-642-39206-1\_5},
  bibsource    = {dblp computer science bibliography, https://dblp.org}
}

@article{BansalGN018,
  author       = {Nikhil Bansal and
                  Shashwat Garg and
                  Jesper Nederlof and
                  Nikhil Vyas},
  title        = {Faster Space-Efficient Algorithms for {S}ubset {S}um, $k$-{S}um, and Related Problems},
  journal      = {{SIAM} J. Comput.},
  volume       = {47},
  number       = {5},
  pages        = {1755--1777},
  year         = {2018},
  url          = {https://doi.org/10.1137/17M1158203},
  doi          = {10.1137/17M1158203},
  bibsource    = {dblp computer science bibliography, https://dblp.org}
}

@inproceedings{DinurDKS12,
  author       = {Itai Dinur and
                  Orr Dunkelman and
                  Nathan Keller and
                  Adi Shamir},
  editor       = {Reihaneh Safavi{-}Naini and
                  Ran Canetti},
  title        = {Efficient Dissection of Composite Problems, with Applications to Cryptanalysis,
                  Knapsacks, and Combinatorial Search Problems},
  booktitle    = {Advances in Cryptology - {CRYPTO} 2012 - 32nd Annual Cryptology Conference,
                  Santa Barbara, CA, USA, August 19-23, 2012. Proceedings},
  series       = {Lecture Notes in Computer Science},
  pages        = {719--740},
  publisher    = {Springer},
  year         = {2012},
  url          = {https://doi.org/10.1007/978-3-642-32009-5\_42},
  doi          = {10.1007/978-3-642-32009-5\_42},
  bibsource    = {dblp computer science bibliography, https://dblp.org}
}

@article{Nederlof26,
  author       = {Jesper Nederlof},
  title        = {An invitation to "Fine-grained complexity of {NP} -complete problems"},
  journal      = {Comput. Sci. Rev.},
  volume       = {61},
  pages        = {100919},
  year         = {2026},
  url          = {https://doi.org/10.1016/j.cosrev.2026.100919},
  doi          = {10.1016/J.COSREV.2026.100919},
  bibsource    = {dblp computer science bibliography, https://dblp.org}
}

@inproceedings{Woeginger04,
  author       = {Gerhard J. Woeginger},
  editor       = {Rodney G. Downey and
                  Michael R. Fellows and
                  Frank K. H. A. Dehne},
  title        = {Space and Time Complexity of Exact Algorithms: Some Open Problems (Invited Talk)},
  booktitle    = {Parameterized and Exact Computation, First International Workshop,
                  {IWPEC} 2004, Bergen, Norway, September 14-17, 2004, Proceedings},
  series       = {Lecture Notes in Computer Science},
  pages        = {281--290},
  publisher    = {Springer},
  year         = {2004},
  url          = {https://doi.org/10.1007/978-3-540-28639-4\_25},
  doi          = {10.1007/978-3-540-28639-4\_25},
  bibsource    = {dblp computer science bibliography, https://dblp.org}
}

@article{FominK13,
  author       = {Fedor V. Fomin and
                  Petteri Kaski},
  title        = {Exact exponential algorithms},
  journal      = {Commun. {ACM}},
  volume       = {56},
  number       = {3},
  pages        = {80--88},
  year         = {2013},
  url          = {https://doi.org/10.1145/2428556.2428575},
  doi          = {10.1145/2428556.2428575},
  bibsource    = {dblp computer science bibliography, https://dblp.org}
}

@article{Bellman62,
  author       = {Richard Bellman},
  title        = {Dynamic Programming Treatment of the {T}ravelling {S}alesman {P}roblem},
  journal      = {J. {ACM}},
  volume       = {9},
  number       = {1},
  pages        = {61--63},
  year         = {1962},
  url          = {https://doi.org/10.1145/321105.321111},
  doi          = {10.1145/321105.321111},
  bibsource    = {dblp computer science bibliography, https://dblp.org}
}

@inproceedings{HeldK61,
  author       = {Michael Held and
                  Richard M. Karp},
  editor       = {Thomas C. Rowan},
  title        = {A dynamic programming approach to sequencing problems},
  booktitle    = {Proceedings of the 16th {ACM} national meeting, {ACM} 1961, {USA}},
  pages        = {71},
  publisher    = {{ACM}},
  year         = {1961},
  url          = {https://doi.org/10.1145/800029.808532},
  doi          = {10.1145/800029.808532},
  bibsource    = {dblp computer science bibliography, https://dblp.org}
}

@article{HorowitzS74,
  author       = {Ellis Horowitz and
                  Sartaj Sahni},
  title        = {Computing Partitions with Applications to the {K}napsack Problem},
  journal      = {J. {ACM}},
  volume       = {21},
  number       = {2},
  pages        = {277--292},
  year         = {1974},
  url          = {https://doi.org/10.1145/321812.321823},
  doi          = {10.1145/321812.321823},
  bibsource    = {dblp computer science bibliography, https://dblp.org}
}

@article{SchroeppelS81,
  author       = {Richard Schroeppel and
                  Adi Shamir},
  title        = {A ${T}={O}(2^{n/2})$, ${S}={O}(2^{n/4})$ Algorithm for Certain {NP}-Complete Problems},
  journal      = {{SIAM} J. Comput.},
  volume       = {10},
  number       = {3},
  pages        = {456--464},
  year         = {1981},
  url          = {https://doi.org/10.1137/0210033},
  doi          = {10.1137/0210033},
  bibsource    = {dblp computer science bibliography, https://dblp.org}
}

@article{Karp82,
  author       = {Richard M. Karp},
  title        = {Dynamic programming meets the principle of inclusion and exclusion},
  journal      = {Oper. Res. Lett.},
  volume       = {1},
  number       = {2},
  pages        = {49--51},
  year         = {1982},
  url          = {https://doi.org/10.1016/0167-6377(82)90044-X},
  doi          = {10.1016/0167-6377(82)90044-X},
  bibsource    = {dblp computer science bibliography, https://dblp.org}
}

@article{Woeginger08,
  author       = {Gerhard J. Woeginger},
  title        = {Open problems around exact algorithms},
  journal      = {Discret. Appl. Math.},
  volume       = {156},
  number       = {3},
  pages        = {397--405},
  year         = {2008},
  url          = {https://doi.org/10.1016/j.dam.2007.03.023},
  doi          = {10.1016/J.DAM.2007.03.023},
  bibsource    = {dblp computer science bibliography, https://dblp.org}
}

@inproceedings{BelovaCKM24,
  author       = {Tatiana Belova and
                  Nikolai Chukhin and
                  Alexander S. Kulikov and
                  Ivan Mihajlin},
  editor       = {Timothy M. Chan and
                  Johannes Fischer and
                  John Iacono and
                  Grzegorz Herman},
  title        = {Improved Space Bounds for {S}ubset {S}um},
  booktitle    = {32nd Annual European Symposium on Algorithms, {ESA} 2024, Royal Holloway,
                  London, United Kingdom, September 2-4, 2024},
  series       = {LIPIcs},
  pages        = {21:1--21:17},
  publisher    = {Schloss Dagstuhl - Leibniz-Zentrum f{\"{u}}r Informatik},
  year         = {2024},
  url          = {https://doi.org/10.4230/LIPIcs.ESA.2024.21},
  doi          = {10.4230/LIPICS.ESA.2024.21},
  bibsource    = {dblp computer science bibliography, https://dblp.org}
}

@article{dividecolor,
author = {Chen, Jianer and Kneis, Joachim and Lu, Songjian and M\"{o}lle, Daniel and Richter, Stefan and Rossmanith, Peter and Sze, Sing-Hoi and Zhang, Fenghui},
title = {Randomized {D}ivide-and-{C}onquer: Improved Path, Matching, and Packing Algorithms},
journal = {SIAM Journal on Computing},
volume = {38},
number = {6},
pages = {2526-2547},
year = {2009},
doi = {10.1137/080716475},
URL = {https://doi.org/10.1137/080716475},
eprint = { https://doi.org/10.1137/080716475}
}

@article{Bjorklund14,
  author       = {Andreas Bj{\"{o}}rklund},
  title        = {Determinant Sums for Undirected {H}amiltonicity},
  journal      = {{SIAM} J. Comput.},
  volume       = {43},
  number       = {1},
  pages        = {280--299},
  year         = {2014},
  url          = {https://doi.org/10.1137/110839229},
  doi          = {10.1137/110839229},
  bibsource    = {dblp computer science bibliography, https://dblp.org}
}

@inproceedings{Nederlof20,
  author       = {Jesper Nederlof},
  editor       = {Konstantin Makarychev and
                  Yury Makarychev and
                  Madhur Tulsiani and
                  Gautam Kamath and
                  Julia Chuzhoy},
  title        = {Bipartite {TSP} in ${O}(1.9999^n)$ time, assuming quadratic
                  time matrix multiplication},
  booktitle    = {Proceedings of the 52nd Annual {ACM} {SIGACT} Symposium on Theory
                  of Computing, {STOC} 2020, Chicago, IL, USA, June 22-26, 2020},
  pages        = {40--53},
  publisher    = {{ACM}},
  year         = {2020},
  url          = {https://doi.org/10.1145/3357713.3384264},
  doi          = {10.1145/3357713.3384264},
  bibsource    = {dblp computer science bibliography, https://dblp.org}
}

@inproceedings{Johnson73,
  author       = {David S. Johnson},
  editor       = {Alfred V. Aho and
                  Allan Borodin and
                  Robert L. Constable and
                  Robert W. Floyd and
                  Michael A. Harrison and
                  Richard M. Karp and
                  H. Raymond Strong},
  title        = {Approximation Algorithms for Combinatorial Problems},
  booktitle    = {Proceedings of the 5th Annual {ACM} Symposium on Theory of Computing,
                  April 30 - May 2, 1973, Austin, Texas, {USA}},
  pages        = {38--49},
  publisher    = {{ACM}},
  year         = {1973},
  url          = {https://doi.org/10.1145/800125.804034},
  doi          = {10.1145/800125.804034},
  bibsource    = {dblp computer science bibliography, https://dblp.org}
}

@article{AlonF85,
  author       = {Noga Alon and
                  Peter Frankl},
  title        = {The maximum number of disjoint pairs in a family of subsets},
  journal      = {Graphs Comb.},
  volume       = {1},
  number       = {1},
  pages        = {13--21},
  year         = {1985},
  url          = {https://doi.org/10.1007/BF02582924},
  doi          = {10.1007/BF02582924},
  bibsource    = {dblp computer science bibliography, https://dblp.org}
}

@article{TalvitieK24,
  author       = {Topi Talvitie and
                  Mikko Koivisto},
  title        = {Approximate Counting of Linear Extensions in Practice},
  journal      = {J. Artif. Intell. Res.},
  volume       = {81},
  pages        = {643--681},
  year         = {2024},
  url          = {https://doi.org/10.1613/jair.1.16374},
  doi          = {10.1613/JAIR.1.16374},
  bibsource    = {dblp computer science bibliography, https://dblp.org}
}

@article{Dallant26,
  author       = {Justin Dallant and
                  L{\'{a}}szl{\'{o}} Kozma},
  title        = {Improved space-time tradeoff for {TSP} via extremal set systems},
  journal      = {CoRR},
  volume       = {abs/2604.05645},
  year         = {2026}
}
\appendix

\section{\TSP in $O^*(4^N)$ time and polynomial space}\label{sec:gs}
The algorithm is a (small) adjustment of the $4^N N^{O(\log N)}$ time algorithm of~\cite{GurevichS87} (see also e.g.~\cite[Section 10.1]{FominK10}) and seems to be folklore. The adjustment that avoids the $N^{O(\log N)}$ term in the running time is to simultaneously compute minimal lengths of paths between all pairs of end points. We assume the \TSP instance is given by a weight function $w$ that outputs for each pair of elements $i,j \in \{1,\ldots,N\}$ a weight $w(i,j)$. Note that the idea presented here was already present in previous work on the \textsc{$k$-Path} problem~\cite{dividecolor}. The pseudocode of the algorithm is as follows:

\begin{algorithm}	
  \caption{TSP in $4^N$ time and $N^{O(1)}$ space.}
  \label{alg:ls}
	\begin{algorithmic}[1]
    \REQUIRE $\mathtt{TSP}(X)$ \hfill\algcomment{$X$ is a subset of cities}
    \ENSURE A table $L$ that stores, for each $s,t \in X$, the minimum length of a $st$-path visiting $X$.
        \IF{$|X| \leq 3$}
            \FOR{distinct $s,t \in X$}
                \STATE Compute and store $L[s,t]$ with brute-force
            \ENDFOR 
        \ENDIF
		\STATE $L[s,t]=\infty$
		\FOR{$Y \subseteq X$, $|Y|=\lceil |X|/2 \rceil $}
			\STATE $L_l = \mathtt{TSP}(Y)$
            \STATE $L_r = \mathtt{TSP}(X \setminus Y)$
            \FOR{distinct $s,u \in Y$ and distinct $v,t \in X \setminus Y$}
                \STATE $L[s,t] = \min \{ L[s,t], L_l[s,u]+w(u,v)+L_r[v,t]\}$
            \ENDFOR
		\ENDFOR
		\STATE \algorithmicreturn \enspace$L$
	\end{algorithmic}
\end{algorithm}

\noindent It is clear that the algorithm is correct since any path from $s$ to $t$ that visits $X$ can be decomposed into a path from $s$ to $u$ that visits $Y \subseteq X$, an edge $(u,v)$ and a path from $v$ to $t$ that visits $X \setminus Y$.

The algorithm clearly uses only polynomial space, and if $T(N)$ is the number of recursive calls of $\mathtt{TSP}(X)$ when $|X|=N$, then we have the recurrence
\[
    T(N) \leq
    \begin{cases}
     1,& \text{if $N \leq 3$}\\
    2 \cdot 2^{N} \cdot T(\lceil N/2\rceil), & \text{otherwise.}\\
    \end{cases}
\]
And it is easily seen that $T(N) \leq O(4^{N})$. Hence Algorithm $\mathtt{TSP}(\{1,\ldots,N\})$ runs in time $O^*(4^N)$.

\section{Omitted Proofs towards Lemma~\ref{lem:bipreg}}\label{sec:MissingProofs}

\begin{proof}[Proof of Claim~\ref{clm:LBOnn0}]

    Note that as \(h(p)=h(1-p)\) then we instead aim for proving that for any positive integer \(d\), and fixed constants \(q\in(0,\frac{1}{2}]\) and \(\varepsilon\in (0,\frac{1}{2})\), there exists \(n_0\) such that if \(n>n_0\), then \[\left(\sum_{i=0}^{n/2} \binom{n/2}{i}2^{n/2\cdot\frac{\binom{n/2-d}{i}}{\binom{n/2}{i}}}\right)^{2/n}\ge 2^{(h(q)+(1-q)^d)(1-\varepsilon)}.\]
    For notational simplicity let \(m=n/2\). Let
    \[m_0:=\max\left\{\frac{d+1}{1-q},\frac{d^2q}{\varepsilon(1-q)},(1/(\varepsilon+h(q)(1-\varepsilon/2)-1))^2+2,\frac{4}{2^{\varepsilon/2}-1}\right\}.\] We show that it is sufficient that if \(m>m_0\) our inequality holds. In order to do this we first prove the following useful claim. Here we use \(\lceil x \rceil\) to denote the ceiling of \(x\).

    \begin{claim}\label{clm:helper}
        For \(m>m_0\) the following inequalities hold:
        \[\binom{m-d}{\lceil mq\rceil}/\binom{m}{\lceil mq\rceil}\ge (1-q)^d(1-\varepsilon),\]
     \[m(1-\varepsilon)h(q)\le \log_2 \binom{m}{\lceil mq\rceil}.\]
    \end{claim}

    \begin{proof}
        We begin by proving the first inequality. Note that as \(m>\frac{d+1}{1-q}\) we have \(m-d>\lceil mq \rceil\). Now, \[\binom{m-d}{\lceil mq\rceil}/\binom{m}{\lceil mq\rceil}\ge \left(\frac{m-d-\lceil mq\rceil}{m-d}\right)^d=\left(1-\frac{\lceil mq \rceil}{m-d}\right)^d\]
        \[>\left(1-\frac{mq}{m-d}\right)^d=\left(1-q-\frac{dq}{m}\right)^d\]
        As \(m>\frac{dq}{(1-q)\varepsilon}\), then 
        \[\left(1-q-\frac{dq}{m}\right)^d>\left((1-q)(1-\frac{\varepsilon}{d})\right)^d\ge (1-q)^d(1-\varepsilon).\]
    
        Now we prove the second inequality. As \(m\) is an integer, then
        \[\binom{m}{\lceil mq\rceil}>\frac{2^{m\cdot h(\lceil mq \rceil/m)}}{m}.\]
        Since \(q\in(0,\frac{1}{2}]\)
        either we have \(\lceil mq\rceil \le m/2\) and then
        \[ h(\lceil mq\rceil/m)> h(q),\]
        or otherwise \(\lceil mq\rceil=\frac{m+1}{2}\) and then
        \[h\left(\lceil mq\rceil/m\right)=h\left(\frac{m+1}{2m}\right)>\log_2\frac{2m}{m+1}=1+\log_2\frac{m}{m+1}>1-\frac{\varepsilon}{2}= h(\frac{1}{2})(1-\varepsilon/2),\]
        where in the last inequality we used the fact that \(m>\frac{4}{-1+2^{\varepsilon/2}}\).
        Therefore in both cases
        \[h(\lceil mq \rceil/m)>h(q)(1-\varepsilon/2).\]
        Now as \(2^{m\cdot h(\lceil mq \rceil/m)}/m=2^{m\cdot h(\lceil mq \rceil/m)-\log_2 m}\) and \(m\ge (1/(\varepsilon+h(q)(1-\varepsilon/2)-1))^2+2\) we have \[m\cdot h(q)(1-\varepsilon/2)-\log_2m>m\cdot h(q)(1-\varepsilon),\]
        and thus 
        \[m(1-\varepsilon)h(q)\le \log_2 \binom{m}{\lceil mq\rceil}.\]
        
    \end{proof}

    Using Claim~\ref{clm:helper}, by setting \(i=\lceil mq\rceil\), we get
    \[\binom{m}{i}2^{m\cdot\frac{\binom{m-d}{i}}{\binom{m}{i}}}\ge 2^{m\cdot (h(q)+(1-q)^d)(1-\varepsilon)}.\]

\end{proof}

\begin{claim}
    For every positive integer \(d\),
    \[\max_{0<q<1}(2^{h(q)+q^d})\cdot{\binom{2d}{d}^{1/2d}}> 3.6.\]
\end{claim}
\begin{proof}
    Note that since \(\binom{2d}{d}^{1/2d}\)is an increasing function for \(d\ge 1\) and since \(\binom{16}{8}^{1/16}\ge 1.807\), then for any \(d\ge 8\),
    \[(2^{h({1/2})+{1/2}^d}) \cdot{\binom{2d}{d}^{1/2d}}\ge  2\cdot 1.807 > 3.61.\]

    For \(d\le 5\), we set $q=7/8$ to derive
    \[\binom{2d}{d}^{1/2d}\cdot2^{(7/8)^d+h({7/8})}=\binom{2d}{d}^{1/2d}\cdot8\cdot 7^{-7/8}\cdot 2^{(7/8)^d}>\binom{2d}{d}^{1/2d}\cdot1.457569\cdot 2^{(7/8)^d}>3.6.\]
   
     For \(d= 6\) and \(d=7\), setting \(q\) to \( 0.9750364898053781\) and \(0.99\) respectively yields\[\binom{2d}{d}^{1/2d}\cdot 2^{h(q)+q^d}>3.6.\]

\end{proof}

\section{Computation of the Number of Ideals and linear Extensions}\label{sec:AppendixImplementation}

In order to compute the number of ideals and linear extensions of the posets described in Section~\ref{sec:con}, we implemented the corresponding algorithms in code. The full implementation is available (anonymously, e.g., via an incognito browser) at this link.\footnote{
\href{https://drive.google.com/drive/folders/1d03qB7IYZbs6C21Z1fqsOCWyv6L_Kdq0?usp=drive_link}{https://drive.google.com/drive/folders/1d03qB7IYZbs6C21Z1fqsOCWyv6L{\textunderscore}Kdq0?usp=drive{\textunderscore}link}}

For the sake of completeness, and to assist the reader in understanding the implementation, we provide in this appendix a high-level description of the algorithms used in the code. 

In the following subsections we discuss the algorithm we used to efficiently compute the number of ideals and linear extensions of \(P\), where \(P\) is a regular bipartite poset as described in Section~\ref{sec:con}.

Let
\[
P=P_{n,D}
\]
be a bipartite poset with bipartition \(X,Y\) where
\[
    X=\{x_0,\dots,x_{n-1}\}\quad  \text{ and } \quad Y=\{y_0,\dots,y_{n-1}\}.
\]
Also \((x_i,y_j)\) is an arc if and only if
\[
(i - j)\bmod n\in D.
\]

\subsection{Counting ideals.}
Because \(P\) is a bipartite poset, every ideal is determined by a subset
\(X'\subseteq X\) together with an arbitrary subset \(Y' \subseteq Y\setminus N(X')\).
As discussed earlier in Section~\ref{sec:ubBiPoset} there are exactly \(2^{|Y\setminus N(X')|}\) ideals \(I\) that have \(X\setminus I=X'\), and hence
\[
\alpha(P)= \sum_{X' \subseteq X}2^{|Y\setminus N(X')|}.
\]
This gives a direct exact algorithm exponential in \(n\).

For larger instances, we use a dynamic programming algorithm which can be formulated as computing the number of cyclic walks\footnote{A cyclic walk of length $n$ in a digraph is a sequence of arcs $a_1,\ldots,a_n$ of $G$ such that the head of $a_i$ equals the tail of $a_{(i+1) \bmod n}$ for each $i$.} of length $n$ in an exponentially sized directed multi-graph $G$. Specifically, assume that $D \subseteq \{0,\ldots,w\}$. For each $i \in \{0,\ldots,n-1\}$ and $W \subseteq D$ we add a vertex $(i,W)$ to $G$, and we add one arc from $(i,W)$ to $((i+1) \bmod\ n,W')$ if
\begin{equation}
 \label{eqsing}
 \{w+1: w \in W\} \setminus \{1\} = W' \setminus \{w\}, \text{ and } W\cap D \neq \emptyset,
\end{equation}
and we add two arcs instead from 
$(i,W)$ to $((i+1) \bmod\ n,W')$ if
\begin{equation}
\label{eqdoubl}
 \{w+1: w \in W\} \setminus \{1\} = W' \setminus \{w\}, \text{ and } W\cap D = \emptyset.
\end{equation}
We claim that the number of cyclic walks of length $n$ of this graph equals the number of ideals of $P$: If condition~\eqref{eqsing} applies, vertex $x_i$ needs to be included in the ideal since we decided that any of its out-neighbors are included. If condition~\eqref{eqdoubl} applies we have the freedom to either include $x_i$ in the ideal or not.
We are interested in cyclic walks since the graph is defined in a modular fashion.

In our code, the number of such cyclic walks is computed in $O(n \cdot 2^w\cdot 2^w)$ time with standard dynamic programming techniques.

\subsection{Exact Computation of \(\lambda(P)\)}

\paragraph{Basic exact recurrence.}
For a subset \(S\subseteq X\), let
\[
e(S)=|\{ y \in Y \mid N(y)\subseteq S\}|.
\]

Let \(F(S,t)\) denote the number of linear extensions whose first \(|S|+t\)
positions contain exactly \(S\) and exactly \(t\) 
elements from \(Y\). Since the poset is bipartite, the next element in such a partial
extension \(\pi\) is either

\begin{itemize}
\item one of the \(n-|S|\) elements \(x_i\) in \(X\setminus S\), or
\item one of the \(e(S)-t\) elements of \(Y\) such as \(y_j\) that is not in \(\pi\) yet but \(N(y_j)\subseteq S\).
\end{itemize}

Therefore
\[
F(S,t)
=
\sum_{x_i\notin S} F(S\cup\{x_i\},t)
+
\bigl(e(S)-t\bigr)F(S,t+1),
\]
with boundary condition
\[
F(X,t)=(n-t)!.
\]
Then the total number of linear extensions of \(P\) is
\[
\lambda(P)=F(\varnothing,0).
\]
Now in order to do this more efficiently we exploit the cyclic symmetry of \(P\).

\paragraph{Cyclic symmetry.}
The poset \(P\) remains the same if one applies a cyclic-shift by \(s\) units on the indices of elements in \(Y\) and \(X\) simultaneously. 
 Hence if \(S\subseteq X\) is rotated by
\(s\), the resulting subset \(\sigma^s(S)\) determines an isomorphic
subproblem, that is,
\[
F(\sigma^s(S),t)=F(S,t)
\]
for all \(S\subseteq X\), \(t\), and \(s\).
Therefore if \(S_1\) can be reached from \(S_2\) by a cyclic shift we say that they are isomorphic and we call a family of subsets of \(X\) a class if it is a maximal family of pairwise isomorphic subsets of \(X\).
We observed that \(F(S,t)\) depends only on the class that \(S\) belongs to. So, in order to get an improved running time we  may compute the recurrence
using only one representative from each class this preserves exactness while significantly reducing the time and space.

\paragraph{Summary.}
In our computations, \(\alpha(P)\) was obtained exactly by the transfer
automaton described above, and for \(\lambda(P)\), we used orbit-compressed dynamic programming. In particular, for \(P\) with \(n=29\) and \(D=(0,1,3,6,10,15)\), the
reported values of \(\alpha(P)\) and \(\lambda(P)\) are computed exactly in seconds.

\subsection{Counting for Irregular Posets}\label{sec:AppendixIrregularPoset}

Let $\mathcal{I}$ be the family of all ideals of $P$. Let $\lambda(I)$ denote the number of linear extensions when it is restricted to the ideal $I \in \mathcal{I}$.
An ideal $I$ is built by adding valid elements one by one, since an element $v \notin I$ can be added to form a new ideal $I \cup \{v\}$ if and only if all its predecessors are already in $I$. 

Equivalently, the number of ways to construct a linear extension over an ideal $I$ is the sum of the ways to construct linear extensions over all valid preceding ideals:
$$\lambda(I) = \sum_{v \in I: v\text{ maximal}} \lambda(I \setminus \{v\}).$$

Based on above observations, we employ a pure dynamic programming
over all possible subsets of vertices to compute the exact number of ideals and linear extensions. This method clearly does not use the structure of the poset to optimize the computation, and thus only works for posets with few vertices.

\end{document}